\documentclass{article}

\usepackage{microtype}
\usepackage{graphicx}
\usepackage{subcaption}
\usepackage{booktabs}
\usepackage{url}
\usepackage{hyperref}

\PassOptionsToPackage{numbers}{natbib}

\usepackage{arxiv_2026}

\usepackage{amsmath}
\usepackage{amssymb}
\usepackage{mathtools}
\usepackage{amsthm}
\usepackage{bm}
\usepackage{enumitem}

\usepackage{algorithm}
\usepackage{algorithmic}

\usepackage[capitalize,noabbrev]{cleveref}

\allowdisplaybreaks

\theoremstyle{plain}
\newtheorem{theorem}{Theorem}[section]

\newtheorem{lemma}[theorem]{Lemma}

\theoremstyle{definition}

\newtheorem{assumption}[theorem]{Assumption}
\theoremstyle{remark}

\newcommand{\R}{\mathbb{R}}
\newcommand{\E}{\mathbb{E}}
\newcommand{\Var}{\mathrm{Var}}
\newcommand{\Cov}{\mathrm{Cov}}

\newcommand{\supp}{\mathrm{supp}}

\newcommand{\x}{\bm{x}}
\newcommand{\y}{\bm{y}}
\newcommand{\z}{\bm{z}}
\newcommand{\h}{\bm{h}}
\newcommand{\s}{\bm{s}}
\newcommand{\eps}{\bm{\varepsilon}}
\newcommand{\A}{\bm{A}}
\newcommand{\J}{\bm{J}}

\newcommand{\Sig}{\bm{\Sigma}}
\newcommand{\Sigeps}{\Sig_\varepsilon}
\newcommand{\Omeg}{\bm{\Omega}}

\title{Inferring Active Neural Circuits Using Diffusion Scores}

\author{
  Savik Kinger \\
  Department of Computer Science\\
  Yale University
  \And
  Johannes Bertram \\
  Machine Learning in Science \\
  Excellence Cluster Machine Learning \\
  University of Tübingen \\
  Tübingen AI Center \\
  ELLIS Institute Tübingen
  \And
  Luciano Dyballa \\
  School of Science \& Technology \\
  IE University 
  \And
  Eviatar Yemini \\
  Department of Neurobiology\\
  University of Massachusetts Chan Medical School 
  \And
  Steven W. Zucker\\
  Depts. of Computer Science and Biomedical Engineering\\
  Wu Tsai Institute\\
  Yale University
}

\begin{document}

\maketitle

\begin{abstract}
\looseness=-1
In biological systems, neural circuits compute through directed, short-latency interactions whose effects unfold across multiple time scales and behavioral contexts. We address the problem of inferring these local, lag-specific interactions from sampled neural population activity under varying stimuli, without assuming a parametric form for the underlying dynamics. Our approach leverages denoising score models by estimating joint-window scores over consecutive activity snapshots (i.e., brain states) and converting these scores into calibrated, directed edge tests via cross-block score products. The key insight is that these products recover the Jacobian of the transition map between brain states under nonlinear dynamics. To cleanly separate lag-specific effects, we introduce minimal multi-block windows that condition on intermediate time points, avoiding the omitted-lag bias inherent in pairwise analyses. The resulting method, \emph{Score--Block Time Graphs} (SBTG), identifies lag-specific directed interactions in sampled neuronal population data. We specifically apply SBTG to whole-brain \emph{C.~elegans} calcium imaging data to recover lag-specific circuit structure not resolved by current methods, including improved alignment with independent connectomes, cell-type-specific temporal organization, and neuromodulatory profiles consistent with known receptor kinetics. These findings highlight the potential for SBTG to serve as a practical ``AI for science'' tool by turning high-dimensional neural population recordings into statistically testable circuit hypotheses.
\end{abstract}

\begin{figure}[ht]
    \centering
    \includegraphics[width=0.95\linewidth]{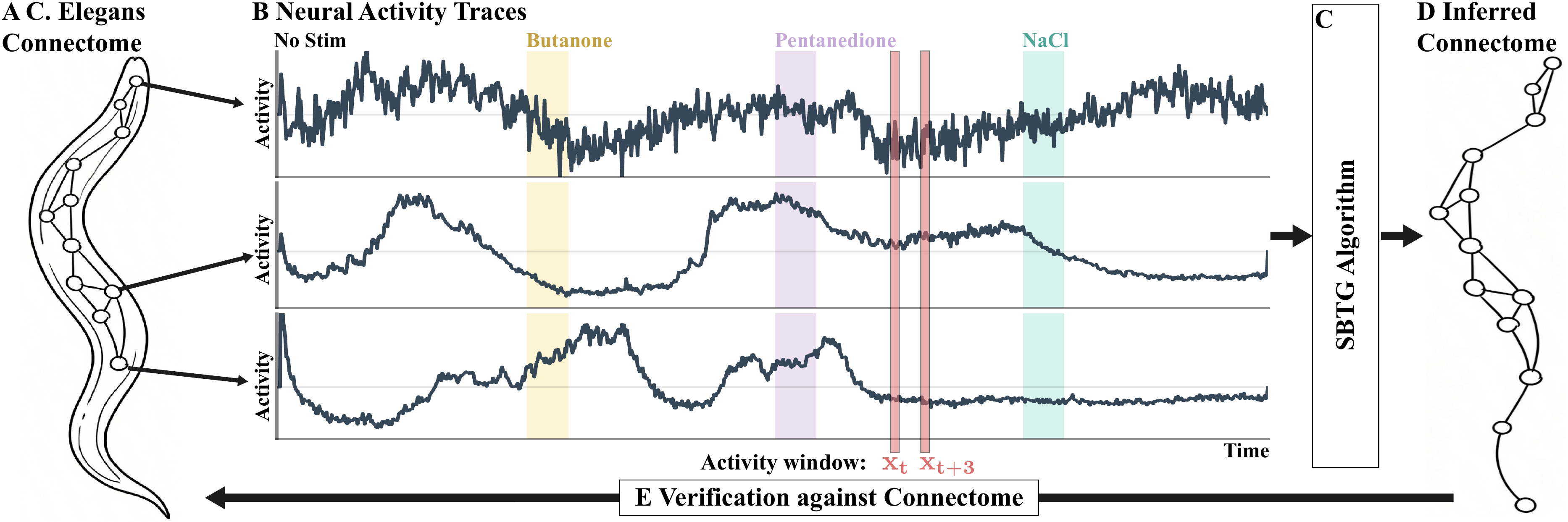}
    \caption{{\bf Goal: inferring connections from activity.} \textbf{A}: \textit{C. elegans} (a small worm) with cartoon anatomical connectome. \textbf{B}: Available data: neural activity measurements for three sampled neurons. The activity window shows observed ``state" of brain at time $t$. \textbf{C}: SBTG algorithm; see Figure~\ref{fig:pipeline}. \textbf{D}: Inferred putative functional connectome; i.e., which connections among neurons likely caused the state to evolve. \textbf{E}: Verification of putative connectome against anatomy.}
    \label{fig:overview}
    \vspace{-2mm}
\end{figure}

\section{Introduction}
\label{sec:intro}
\vspace{-1mm}

\looseness=-1
Neural circuits consist of directed interactions and modulators operating across multiple time scales. A sensory neuron responds to a stimulus, triggering activity in downstream interneurons within tens to hundreds of milliseconds, which in turn drive motor outputs that shape behavior. This directed flow of information---who influences whom, and on what timescale---is the computational foundation of neural processing. Modern high-throughput recording technologies, such as whole-brain calcium imaging, capture the simultaneous activity of tens to thousands of neurons in behaving animals~\citep{Siegle2021,Semedo2019, susoy2021natural}, affording a dynamic view of the functional activity underlying sensation and behavior \cite{flavell2022dynamic}. Yet these recordings capture only the observable outputs of the circuit, not the synaptic wiring diagrams that generate them, nor the wireless extrasynaptic signaling networks (e.g., neuromodulation via monoamines and neuropeptides) that are absent from anatomical connectomes \cite{bargmann2012beyond, bargmann2013connectome, venkatesh2025c, bentley2016multilayer, ripoll2023neuropeptidergic}. Recent work has shown a disconnect between static connectomes and dynamic functional activity measured \textit{in vivo}, often revealing weak or inconsistent correspondence between these two circuit views \cite{Yemini2021, susoy2021natural, uzel2022set, Randi2023, currier2025infrequent}. Understanding the origin of this disconnect is critical. We therefore ask: \emph{can directed, time-resolved circuit structure be recovered from population activity without prior anatomical knowledge?} (\cref{fig:overview}). The answer speaks to the emerging ``AI for science'' objective: using modern machine learning to extract structured, falsifiable scientific conjectures from complex measurements.

\looseness=-1
This problem has attracted substantial methodological effort, but existing approaches face fundamental limitations: nonlinear dynamics, high-dimensional populations, and short recordings with strong temporal dependence. Vector autoregressive models and Granger causality require linearity assumptions that break down for saturating or state-dependent neural responses~\citep{Granger1969,Lutkepohl2005}. Point-process models strongly assume spike observations rather than calcium fluorescence~\citep{Hawkes1971}. Constraint-based causal discovery methods like PCMCI$^+$ handle temporal dependence but rely on conditional-independence testing that limits scalability~\citep{Runge2020PCMCIplus}. Meanwhile, score-based generative models have shown strong empirical alignment with neural population structure---yet their internals remain opaque: they do not reveal \emph{which edges} carry the directed influence~\citep{Ozcelik2023,Kapoor2024}.

\looseness=-1
We build on a different perspective: the local geometry of joint distributions over consecutive activity snapshots. The score function---the gradient of the log-density---encodes how local probability mass evolves as one moves through neural activity space. For a window of observations $(\x_t, \x_{t+1})$, the \emph{cross-block} score product asks how the score with respect to the future state depends on perturbations of the past state. This quantity is intimately related to the mixed Hessian of the log-density, which---under additive-noise dynamics---corresponds to the Jacobian of the underlying state transition map. The Jacobian is precisely the object of interest: its $(j,i)$ entry quantifies how strongly neuron $i$ at time $t$ appears to influence neuron $j$ at time $t{+}1$. Cross-block score products thus provide a window into directed causal structure, estimable from a learned score model without assuming any parametric form for the dynamics.

A critical challenge arises when interactions span multiple timescales, as is common in biological neural circuits. In animal brains, fast synaptic transmission operates in tens of milliseconds while slow neuromodulatory signals unfold over longer periods \cite{greengard2001neurobiology, bargmann2012beyond, bargmann2013connectome, watteyne2024neuropeptide}. Naively estimating coupling from pairs $(\x_t, \x_{t+1})$ conflates these effects: the inferred ``lag-1'' coupling mixes direct lag-1 influence with indirect effects mediated through unconditioned intermediate time points. This reflects both multi-step links in the connectome \cite{Randi2023, dvali2025diverging, creamer2025bridging} and the effect of diffusive neuromodulators \cite{bargmann2012beyond, bargmann2013connectome, bentley2016multilayer, watteyne2024neuropeptide, ripoll2023neuropeptidergic}.  Formally, this is the multi-lag analog of omitted variable bias in regression, and it means that pair-window methods cannot cleanly separate timescale-specific causal effects. Our solution is to construct \emph{minimal multi-block windows} that condition on all intermediate time points, separating the contribution of each timescale. For lag $\ell$, we model the joint distribution of $(\x_t, \x_{t+1}, \ldots, \x_{t+\ell})$ and extract the cross-block score product between the first and last blocks. This yields per-lag Jacobian estimates that are identifiable under a nonlinear multi-lagged model of brain state dynamics.

\begin{figure}[h]
    \centering
    \includegraphics[width=0.75\linewidth]{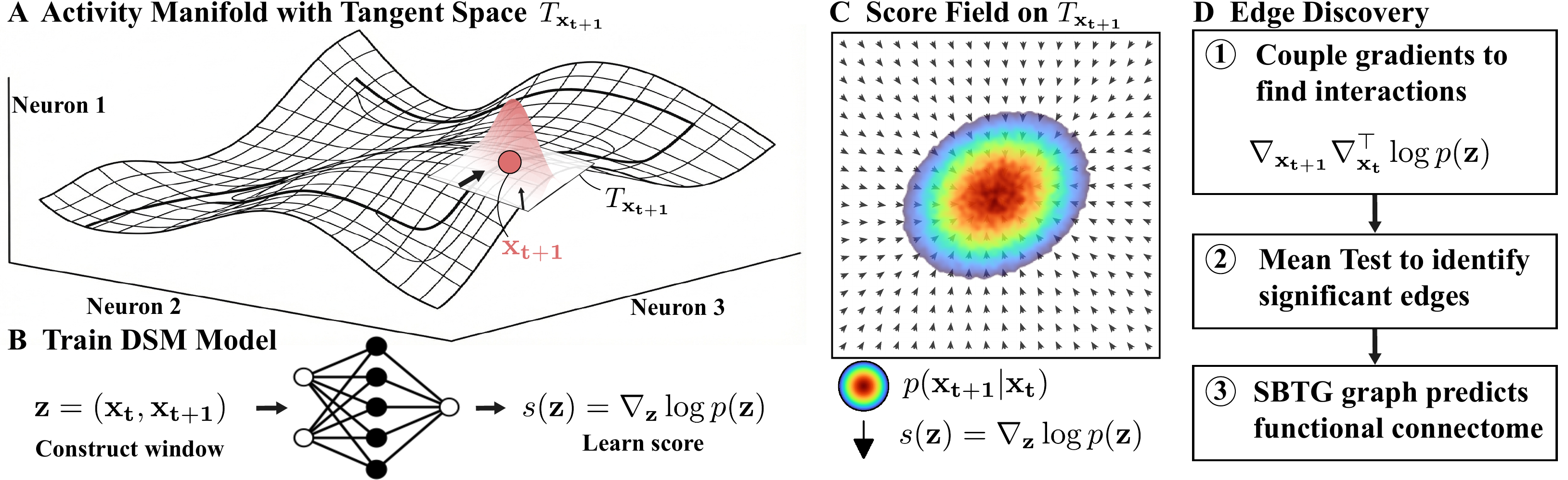}
    \caption{{\bf Learning connections between states.} \textbf{A}: States of the brain evolve as dynamics on some manifold in neural space. We model the transitions $\mathbf{z}$; i.e., the tangents to these dynamics (tangent space $T_{\mathbf{x_{t+1}}}$ shown). \textbf{B}: Training DSM model on activity transition window to learn scores. \textbf{C}: Cartoon score field on $T_{\mathbf{x_{t+1}}}$, showing the score pointing towards higher probability transitions. \textbf{D}: Steps to infer likely connections.}
    \label{fig:pipeline}
    \vspace{-2mm}
\end{figure}

We call the resulting method \emph{Score--Block Time Graphs} (SBTG), reflecting its construction from joint-window score geometry and its output as lag-specific directed graphs. SBTG incorporates several design choices critical for real neural data: null-contrast hyperparameter selection that avoids circularity with ground truth, cross-fitting for valid out-of-sample inference, and heteroskedasticity and autocorrelation consistent (HAC) standard errors for temporal dependence. The output is a set of lag-specific directed graphs with calibrated significance and sign information.

Whole-brain \emph{C.~elegans} calcium imaging data provides a particularly well-aligned and controlled setting for this approach. \emph{C.~elegans'} compact nervous system and stereotyped wiring reduce the severity of hidden-neuron confounding and support pooling partially overlapping recordings under a shared-dynamics assumption, while independent connectome atlases provide directed references for downstream validation. We therefore make the following contributions:
\vspace{-1mm}
\begin{enumerate}[leftmargin=1.2em]
  \item \textbf{Approach.} We develop Score--Block Time Graphs (SBTG), a method that provably recovers lag-specific Jacobians (in expectation) from cross-block score products under additive-noise dynamics (\cref{thm:jacobian_main}). 
  \item \textbf{Dynamic connectome identification.} Applying SBTG to whole-brain \emph{C.~elegans} calcium imaging, we recover lag-specific directed structure that aligns with independent connectomes and reveals biologically meaningful organization across time scales.
\end{enumerate}

\section{Related Work}

\textbf{Neural Population Connectivity.} Large-scale extracellular probes and dense optical imaging have shifted systems neuroscience from single-unit tuning to population interactions, motivating methods that respect both the local geometry and the directed temporal structure of cortical dynamics. Empirical work suggests that information exchange across areas is confined to low-dimensional ``communication subspaces,'' often expressed over specific temporal delays, so that effective coupling is expressed in constrained directions of population activity rather than uniformly across neurons~\citep{Semedo2019}. Meanwhile, surveys in mouse visual cortex show functional hierarchies consistent with anatomical pathways~\citep{Siegle2021}, and classic EM-plus-physiology studies link response tuning to local microcircuit wiring~\citep{Bock2011}. This motivates learning \emph{time-directed}, lag-specific interactions rather than static correlations. However, most existing population-level approaches either summarize interactions at a single timescale or rely on parametric dynamical assumptions, limiting their ability to resolve multi-timescale circuit structure.

\textbf{Time-Series Structure Learning.} Classical approaches include linear VAR models and Granger causality, which are well understood but struggle when effects are nonlinear and conditioning sets are large~\citep{Granger1969,Lutkepohl2005}. Point-process models such as Hawkes processes capture spike-train dynamics yet impose strong parametric assumptions and require careful regularization in high dimension~\citep{Hawkes1971}. Constraint-based discovery methods adapted to time series (e.g., PCMCI$^+$) explicitly handle autocorrelation and reduce conditioning sets, but still rely on conditional-independence testing and can face power/computation challenges as variables and lags grow~\citep{Runge2020PCMCIplus}. Multivariate transfer entropy can capture nonlinear, lagged dependencies by measuring the predictive information a source adds conditional on a target, but fails to directly correspond to a differential effect under a nonlinear, additive noise model~\citep{Novelli2019}. DAG optimization methods and time-series variants (e.g., Dynotears) are effective under restrictive assumptions but are not designed for high-dimensional nonlinear population dynamics without strong priors~\citep{pamfil2020dynotearsstructurelearningtimeseries,Zheng2018,Yu2019}. Recent deep-learning approaches frame connectivity inference as latent-variable graph inference (NRI~\citep{Kipf2018NRI}), low-rank dynamical reconstruction (LINT~\citep{Valente2022LINT}), or transformer-inspired recovery (NetFormer~\citep{Lu2025NetFormer}); we benchmark against all three on the empirical \emph{C.~elegans} setting in Appendix~\ref{app:dl_baselines}.

\looseness=-1
\textbf{Score-Based Modeling for Structure.} Denoising score matching and diffusion-style models estimate the score $\nabla \log p(x)$, enabling scalable modeling without normalized likelihoods~\citep{Hyvarinen2005,Vincent2011,SongErmon2019,Song2021}. Beyond images, these models are increasingly aligned with neural data, motivating the use of scores as a representation of local geometry of neural window distributions~\citep{Ozcelik2023,Kapoor2024}. Recent score-based causal discovery approaches target i.i.d.\ settings or specific families of perturbed systems~\citep{varici2023,MontagnaBerrett2023,Zhu2023}. We depart from this by targeting lagged, multi-stimulus neural time series via joint-window score geometry and by converting scores into statistically \emph{calibrated, directed} edge tests with lag interpretation.

\section{Problem Statement}
\label{sec:problem}

We now formalize the inference problem, introducing multi-lag neural dynamics, the role of the coupling Hessian, and the practical challenges posed by partial observability across multiple recordings.

\textbf{Multi-Lag Neural Dynamics.} We model population neural activity as a discrete-time stochastic process observed at a fixed sampling rate. Let $\x_t \in \R^n$ denote the activity of $n$ neurons at time $t$. In our \emph{C.~elegans} data, the sampling rate is 4 Hz (4 3D-image volumes/second), so each time step corresponds to 250 ms---a timescale relevant for both fast chemical synapses and slower modulatory effects. Our central modeling assumption is that the future state depends on a finite history of past states through a nonlinear transition map with additive noise $\x_{t+1} = g(\x_t, \x_{t-1}, \ldots, \x_{t-L+1}) + \eps_t,$
where $g: \R^{Ln} \to \R^n$ is a differentiable transition function encoding the circuit's computational structure, $L$ is the lag order, and $\eps_t \sim \mathcal{N}(0, \Sigeps)$ represents independent stochastic neural variability.

The directed circuit structure at each lag is encoded in the Jacobian blocks of the transition map:
\begin{equation}
\J_\ell := \frac{\partial g}{\partial \x_{t+1-\ell}} \in \R^{n \times n},
\end{equation}
where entry $(\J_\ell)_{ji}$ quantifies the local, instantaneous influence of neuron $i$ at time $t{-}\ell{+}1$ on neuron $j$ at time $t{+}1$. We aim to recover the support and sign pattern of the lag-specific Jacobians from observed time series, identifying which directed interactions are active and at what timescales.

\textbf{The Omitted-Lag Problem.} A natural first approach is to analyze pairs of consecutive observations $(\x_t, \x_{t+1})$. However, when $L \ge 2$, this pair-window approach suffers from an \emph{omitted-lag bias}: the inferred coupling conflates multiple lag effects through the autocorrelation structure of the process. To see this concretely, consider a scalar VAR(2) process:
$
x_{t+1} = a_1 x_t + a_2 x_{t-1} + \varepsilon_t.
$
From pairs $(x_t, x_{t+1})$ alone, the best linear predictor is $\E[x_{t+1}|x_t] = (a_1 + a_2 b)x_t$ where $b = \Cov(x_{t-1}, x_t)/\Var(x_t)$. The reduced-form coefficient $a_1 + a_2 b$ mixes both lag effects---we cannot separately identify $a_1$ from pairs alone. This reduced-form mixing of effects generalizes to the multivariate nonlinear case: pairs cannot separate lag-specific causal effects. Our multi-block window construction directly addresses this issue by conditioning on all intermediate time points (i.e., we can then analyze $\E[x_{t+1}|x_t, x_{t-1}] = a_1x_t + a_2 x_{t-1}$), enabling clean identification of each $\J_\ell$.

\textbf{The Coupling Hessian: Geometric Intuition.} For a window $\z = (\x_t, \x_{t+1})$, the score function $\s(\z) = \nabla_{\z} \log p(\z)$ describes how log-probability changes as we move in window-space. Geometrically, the score points toward higher-density regions; its magnitude reflects the local curvature of the probability landscape. Now consider the cross-block component for window $\z = (\x_t, \x_{t+1})$: how does the score with respect to the future $\x_{t+1}$ depend on the past $\x_t$? This dependence is captured by the mixed Hessian $\nabla_{\x_{t+1}} \nabla_{\x_t}^\top \log p(\z)$---a measure of how the probability surface ``tilts'' between past and future. Steeper tilt means stronger influence. And concretely, we have that $\nabla_{\x_{t+1}} \nabla_{\x_t}^\top \log p(\x_t, \x_{t+1}) = \nabla_{\x_{t+1}} \nabla_{\x_t}^\top \log p(\x_{t+1} | \x_t)$, since the $\log p(\x_t)$ term vanishes under $\nabla_{\x_{t+1}}$. This highlights that the cross-block curvature is determined by the local transition structure.

By the Score-Hessian Identity, score products allow us to estimate this cross-block curvature:
\begin{equation}
\E[\s_{t+1}(\z) \, \s_t(\z)^\top] = -\E[\nabla_{\x_{t+1}} \nabla_{\x_t}^\top \log p(\z)].
\label{eq:score_hess_main_text}
\end{equation}
Under the additive-noise dynamics 
this mixed Hessian equals a noise-scaled Jacobian, $\Sigeps^{-1} \J_1$, so that the cross-block score product recovers $-\Sigeps^{-1}\J_1$ in expectation via \eqref{eq:score_hess_main_text}. Accordingly, we interpret signs using the sign-aligned coupling $-\E[\s_{t+1}(\z)\s_t(\z)^\top]$.

\textbf{Multiple Recordings and Partial Observability.} Neural recordings rarely observe all neurons. In \emph{C.~elegans}, each recording observes a different subset $O_u \subset [N]$ of the canonical neuron set due to imaging geometry and neural identification challenges. Across our 21 recordings, individual neurons appear in 12--20 worms (median 18 worms per neuron), and only 189 head neurons and 42 tail neurons of the 302 canonical neurons are reliably identified \cite{Yemini2021}. This partial overlap creates both a statistical challenge (how to aggregate evidence across recordings) and an identifiability question (what can we recover).

When the observed neurons $O$ are a proper subset of the full population, any inferred coupling is necessarily a reduced-form quantity that reflects both direct effects and mediation through hidden neurons. Formally, let $\tilde{g}_O$ denote the effective transition function governing the observed marginal dynamics. SBTG identifies $\E[\J_{\tilde{g}_O, \ell}]$---the expected Jacobian of this reduced-form map. In \cref{thm:unified_new}, we show that this decomposes into the conditional expectation of the true Jacobian plus a bias term that depends on how hidden neurons correlate with observed dynamics. 
For \emph{C.~elegans}, this bias is substantially mitigated by the organism's small nervous system: with only 302 neurons total and our coverage of 189 reliably identified head neurons, we observe a substantial fraction of the circuit. Moreover, the stereotyped wiring of \emph{C.~elegans}---documented by electron microscopy reconstructions~\citep{white1986structure, Cook2019, witvliet2021connectomes}---means that the different recordings provide complementary views of the same underlying circuit. Under this shared-dynamics assumption, aggregating evidence across recordings yields a consistent estimator of population-level coupling.

\section{Methods: Score--Block Time Graphs}
\label{sec:method}

The preceding section identifies what SBTG should recover (lag-specific Jacobians) and what can go wrong (omitted-lag bias, hidden-neuron confounding). We now describe how to translate these insights into a practical pipeline. The core workflow consists of learning a score model for minimal multi-block windows, computing cross-block score products, and testing whether these products are significantly nonzero after correcting for temporal dependence. Several design choices are critical, and we motivate them below.

\textbf{Minimal Multi-Block Windows.} Separation of lag-specific effects requires conditioning on all intermediate time points. For each target lag $\ell$, we therefore construct a minimal multi-block window: $\z^{(\ell)}_t = (\x_t, \x_{t+1}, \ldots, \x_{t+\ell}) \in \R^{(\ell+1)n}.$
This construction conditions on all intermediate time points while allowing per-lag hyperparameter tuning and model selection. The window dimension grows linearly with lag ($2n$ for $\ell = 1$, $3n$ for $\ell = 2$, etc.), inducing an explicit tradeoff between statistical efficiency and conditioning completeness. In short, minimal multi-block windows provide clean lag separation without unnecessary dimensional growth.


\textbf{Structured Score Model.} Given these windows, we require a score model that captures the joint distribution's local geometry while remaining computationally tractable in high dimension. We parameterize the score through an energy function that decomposes into within-block and cross-block terms. Write the window as $\z_t = (\z^{(0)}_t,\z^{(1)}_t,\ldots,\z^{(\ell)}_t)$ where $\z^{(k)}_t := \x_{t+k}\in\R^n$ denotes the $k$-th time block. We use:
\begin{equation}
U_\theta(\z_t) = \sum_{k=0}^{\ell} g_k(\z^{(k)}_t) + \sum_{r=1}^{\ell} (\z^{(\ell)}_t)^\top W_r \, \z^{(\ell-r)}_t,
\label{eq:score_model}
\end{equation}
where $g_k$ are multi-layer perceptrons capturing within-block structure and $W_r \in \R^{n \times n}$ are explicit coupling matrices between time blocks. The score is $\widehat\s_\theta(\z_t) = -\nabla_{\z_t} U_\theta(\z_t)$. This structured form balances expressiveness with interpretability and statistical efficiency: the MLPs can capture complex marginal structure within each time block, while the coupling matrices $W_r$ force the DSM objective to learn the correct cross-block influence. This is a deliberate bias--variance choice: the model class is more constrained than a fully generic score network, but it is better aligned with the downstream goal of recovering interpretable lag-specific interactions in data-limited recordings. In particular, the cross-partial $\partial^2 U_\theta / \partial \z^{(\ell-r)}_t \, \partial \z^{(\ell)}_t = W_r$, so each explicit coupling matrix $W_r$ acts directly as a directed adjacency from time block $\ell{-}r$ to time block $\ell$, rather than being implicit in the weights of a generic black-box network.
Training minimizes the denoising score matching objective:
\begin{equation}
\mathcal{L}_{\mathrm{DSM}}(\theta) = \E\left[\left\|\widehat\s_\theta(\z + \sigma\eps) + \eps/\sigma\right\|^2\right],
\end{equation}
where $\eps \sim \mathcal{N}(0, I)$ and $\sigma$ is a noise level hyperparameter. This objective coincides with the score matching objective used in state-of-the-art generative diffusion models. Intuitively, the target $-\eps/\sigma$ is the score of the Gaussian corruption kernel applied to $\z_t$; the denoising identity~\citep{Vincent2011} then guarantees that minimizing this loss estimates the score of the underlying joint-window distribution of neural activity, without ever requiring its normalizing constant.

\textbf{Null-Contrast Hyperparameter Selection.} Hyperparameter selection for structure learning is subtle, particularly for high-dimensional biological neural time series. The DSM validation loss measures score-function accuracy but does not correlate with edge-recovery performance: a model may achieve low DSM loss while learning spurious correlations that do not correspond to directed circuit interactions. Using ground-truth edges for tuning would introduce circularity and defeat the purpose of the model. We address this by tuning to maximize the ``null contrast'':
\begin{equation}
\mathrm{NC} := \frac{\mathrm{mean}(|\widehat\mu_\ell|_{j \neq i})}{\mathrm{mean}(|\widehat\mu_\ell^{\mathrm{null}}|_{j \neq i})},
\label{eq:null_contrast}
\end{equation}
where $\widehat\mu_{\ell,ji} = \frac{1}{T_{\mathrm{win}}} \sum_{t} \widehat\s_{t+\ell,j}(\z^{(\ell)}_t) \cdot \widehat\s_{t,i}(\z^{(\ell)}_t)$ is the lag-$\ell$ cross-block score product (matrix entry $(j,i)$, averaged over windows), and $\widehat\mu_\ell^{\mathrm{null}}$ is the same statistic computed after temporally permuting lag-block scores, breaking the true joint structure while preserving marginal properties. Higher null contrast indicates a stronger signal relative to a distribution-breaking null, without requiring edge labels. This objective rewards hyperparameter settings that find notably reproducible structure---structure that we then validate externally against known knowledge of the synaptic and neuromodulator connectomes \cite{white1986structure, Cook2019, bentley2016multilayer}. 

\textbf{Inference under Temporal Dependence.} To avoid overfitting bias, we use 5-fold cross-fitting: the score model $\widehat\s_\theta$ is trained on 80\% of windows, and all edge statistics reported are computed on held-out data. However, temporal dependence introduces a second inference challenge. The product series $Y_t := \widehat\s_{t+\ell,j}(\z^{(\ell)}_t) \cdot \widehat\s_{t,i}(\z^{(\ell)}_t)$ inherits autocorrelation from the overlapping window structure. Standard errors that assume independent observations would overstate significance---in our data, by a factor of three. We use Newey--West HAC estimators with bandwidth calibrated to window size and sampling rate, incorporating lagged autocovariances:
$\widehat{\sigma^2_{\mathrm{NW}}} = \widehat\gamma_0 + 2\sum_{h=1}^m \left(1 - \frac{h}{m+1}\right) \widehat\gamma_h$
where $\widehat\gamma_h$ is the lag-$h$ sample autocovariance and bandwidth $m$ controls how many lags to include.\cite{NeweyWest1987}

With $n=80$ neurons, we test $n(n-1)=6{,}320$ directed edges per lag, making significance testing essential.
We use Benjamini--Yekutieli (BY) because it provides finite-sample False Discovery Rate (FDR) control under arbitrary dependence among test statistics~\citep{BenjaminiYekutieli2001}. Since our $p$-values derive from temporally overlapping windows and a shared score model, they exhibit unknown dependence structures. The BY test provides principled FDR control under overlapping windows. We conduct a sensitivity analysis to arrive at our chosen sensitivity of $\alpha = 0.10$ in Figure~\ref{app:sensitivity-analysis}. Together, cross-fitting and HAC inference yield calibrated edge statistics despite temporal dependence. The complete SBTG procedure is summarized as pseudocode in Algorithm~\ref{alg:sbtg} (\cref{app:implementation}).

\label{sec:theory}
\textbf{Theoretical Guarantees.} We now state the core identification results; the complete formal treatment with explicit assumptions and proofs appears in \cref{app:model}. Under the additive-noise dynamics
\begin{equation}
\y_{t+1} = f(\y_t, \ldots, \y_{t-L+1}) + \bm{\eta}_t,
\qquad
\bm{\eta}_t \sim \mathcal{N}(0, \bm{\Omega}),
\end{equation}
we show that cross-block score products computed from joint-window score models recover lag-specific Jacobian structure in expectation.

Fix a lag $\ell \in \{1,\dots,L\}$ and define the minimal multi-block window $\z_t := (\y_t,\y_{t+1},\ldots,\y_{t+\ell})$. Let $\s_\tau(\z_t) := \nabla_{\y_\tau}\log p(\z_t)$ denote the score component with respect to the block at time $\tau$. Our central summary statistic is the cross-block score product matrix
\begin{equation}
\bm{M}_\ell
\;:=\;
\E\!\left[\s_{t+\ell}(\z_t)\, \s_t(\z_t)^\top\right]
\in \R^{N\times N}.
\end{equation}
The following theorem states that $\bm{M}_\ell$ is proportional to the lag-$\ell$ Jacobian of the underlying transition map, up to noise scaling.

\begin{theorem}[Lag-Specific Jacobian Recovery]
\label{thm:jacobian_main}
For $\ell \in \{1,\dots,L\}$, let $\bm{F}_\ell := \frac{\partial f}{\partial \y_{t+1-\ell}} \in \R^{N\times N}$ denote the lag-$\ell$ Jacobian of the one-step transition map. Then the cross-block score product matrix satisfies
\begin{equation}
\bm{M}_\ell \;=\; -\,\bm{\Omega}^{-1}\,\E[\bm{F}_\ell].
\label{eq:jacobian_main}
\end{equation}
\end{theorem}

Equation~\eqref{eq:jacobian_main} formalizes the central link used by SBTG: score geometry of joint windows reveals directed coupling. The proportionality to $\bm{\Omega}^{-1}$ implies that neurons with larger innovation variance contribute weaker score signals. In particular, when $\bm{\Omega}$ is diagonal, $\bm{\Omega}^{-1}$ acts as a positive row-wise rescaling, preserving the sparsity pattern. 
Moreover, \eqref{eq:jacobian_main} reduces to a global sign flip between $\bm{M}_\ell$ and the average Jacobian. Accordingly, we work throughout with the sign-aligned coupling matrix
$\bm{C}_\ell \;:=\; -\,\bm{M}_\ell$
so that $\mathrm{sign}((\bm{C}_\ell)_{ji}) = \mathrm{sign}(\E[(\bm{F}_\ell)_{ji}])$ under diagonal noise, matching the intended directed effects interpretation.

The structured estimator of Eq.~\eqref{eq:score_model} is a practical realization of the population object that Theorem~\ref{thm:jacobian_main} identifies; we elaborate this distinction in Appendix~\ref{app:pop_vs_est}. Further, when only a subset $O\subset[N]$ of neurons is observed, the same construction applied to the marginal window distribution recovers observed-to-observed coupling up to an explicit bias term induced by marginalizing hidden neurons. \Cref{thm:unified_new} in \cref{app:model} shows that the observed cross-block score product decomposes into a noise-scaled $\E[\bm{F}_\ell[O,O]]$ term plus a bias matrix that depends on conditional uncertainty in the hidden activity given the observed window. This bias vanishes under standard conditional-independence conditions on the hidden components, and in \emph{C.~elegans} is empirically mitigated by the relatively high fraction of neurons observed across recordings.

\section{Results}
\label{sec:results}

We validate SBTG in synthetic settings and in whole-brain \emph{C.~elegans} calcium imaging. We first assess lag-1 performance against anatomical and functional benchmarks, then demonstrate that SBTG recovers biologically plausible multi-timescale structure not captured by standard methods.

\textbf{Synthetic Validation.} We first verify SBTG’s ability to recover known structure on synthetic data. For instance, we generate nonlinear VAR(2) dynamics as follows:
\begin{equation*}
\x_{t+1} = \tanh(\A_1 \x_t + \A_2 \x_{t-1}) + \eps_t, \quad \eps_t \sim \mathcal{N}(0, 0.1 \cdot I),
\end{equation*}
with $n = 20$ neurons, 10\% edge density, and $T = 5000$ time points in Table~\ref{tab:synthetic}. We benchmark against  linear VAR and DYNOTEARS, as well as nonlinear (PCMCI+) causal discovery methods. We include additional synthetic analyses in the appendix~\ref{app:synthetic_full} for Hawkes, VAR, and Poisson processes, and report a larger-scale ($n{=}80$) extension of the tanh benchmark in Table~\ref{tab:res_tanh_n80}, with VAR and Poisson scaling diagnostics described alongside it (Appendix~\ref{app:n80_scaling}).
SBTG achieves substantially higher classification performance (AUROC 0.83) compared to baselines, most of which hover near chance (0.50--0.60). Notably, even nonlinear methods like Dynotears (AUROC 0.57) and PCMCI+ (0.54) struggle to recover the true graph structure in this regime. This supports our claim that SBTG's score-based estimators effectively capture the nonlinear dynamics where standard approaches fail.

\begin{table}[t]
\caption{\textbf{Synthetic benchmark (Nonlinear Tanh).} SBTG outperforms linear and nonlinear baselines in both binary detection (AUROC, AUPRC) and edge recovery (F1).}
\label{tab:synthetic}

\caption{\textbf{Benchmark performance.} Performance against the Cook structural connectome \citep{Cook2019} and Randi functional atlas \citep{Randi2023}.
AUROC measures discrimination of reference edges; $\rho$ denotes Spearman correlation between inferred scores and reference connection weights. 
SBTG has the highest Spearman correlation on both benchmarks, with Cook AUROC similar to Pearson.}
\label{tab:lag1_baselines}

\centering
\begin{minipage}[t]{0.48\linewidth}
\centering
\begin{scriptsize}
\setlength{\tabcolsep}{3pt}
\begin{tabular}{lccc}
\toprule
Method & F1 Score & AUROC & AUPRC \\
\midrule
\textbf{SBTG} & \textbf{0.39 $\pm$ 0.08} & \textbf{0.83 $\pm$ 0.07} & \textbf{0.24 $\pm$ 0.06} \\
DYNOTEARS & 0.23 $\pm$ 0.11 & 0.57 $\pm$ 0.06 & 0.21 $\pm$ 0.09 \\
VAR-LASSO & 0.21 $\pm$ 0.03 & 0.60 $\pm$ 0.05 & \textbf{0.24 $\pm$ 0.08} \\
PCMCI+ & 0.17 $\pm$ 0.03 & 0.54 $\pm$ 0.05 & 0.10 $\pm$ 0.01 \\
VAR-Ridge & 0.16 $\pm$ 0.02 & 0.56 $\pm$ 0.05 & \textbf{0.24 $\pm$ 0.06} \\
VAR-LiNGAM & 0.08 $\pm$ 0.07 & 0.48 $\pm$ 0.04 & 0.09 $\pm$ 0.01 \\
\bottomrule
\end{tabular}
\end{scriptsize}

\end{minipage}\hfill
\begin{minipage}[t]{0.48\linewidth}
\centering
\begin{scriptsize}
\setlength{\tabcolsep}{3pt}
\begin{tabular}{l cccc}
\toprule
& \multicolumn{2}{c}{\textbf{Cook}} & \multicolumn{2}{c}{\textbf{Randi}} \\
\cmidrule(lr){2-3} \cmidrule(lr){4-5}
Method & AUROC & $\rho$ & AUROC & $\rho$ \\
\midrule
SBTG              & \textbf{0.581} & \textbf{0.202} & \textbf{0.637} & \textbf{0.147} \\
Pearson           & 0.576          & 0.106          & 0.596          & 0.098          \\
Cross Correlation & 0.568          & 0.094          & 0.571          & 0.087          \\
Granger           & 0.548          & 0.066          & 0.602          & 0.098          \\
Glasso            & 0.516          & 0.076          & 0.533          & 0.070          \\
\bottomrule
\end{tabular}
\end{scriptsize}

\end{minipage}
\vspace{-4mm}
\end{table}

\looseness=-1
\textbf{\textit{C.~elegans} Data and Structural Benchmarks.} We evaluate on whole-brain calcium imaging from NeuroPAL-annotated \emph{C.~elegans}~\citep{Yemini2021}: 240-second videos at 4 Hz presenting three sensory stimuli. After preprocessing (dF/F$_0$ normalization, $z$-scoring, and combining left--right neuron pairs), each recording contains $n = 80$ identified neurons. We benchmark inferred edges against three independent references: a structural connectome from electron microscopy~\citep{Cook2019}, a functional atlas from optogenetic perturbation~\citep{Randi2023}, and a monoamine connectome from gene expression~\citep{bentley2016multilayer}. We focus first on lag-1 (250 ms) predictions, which are expected to align most closely with monosynaptic transmission. \Cref{tab:lag1_baselines} reports performance across methods.

We observe that SBTG performs competitively across metrics, but in particular on Spearman correlation (0.202 vs. 0.106), indicating better recovery of connection strengths---the relative importance of edges. However, the performance on binary classification metrics reveals that in aggregate, we are reproducing neither the entire anatomical nor functional connectome. We therefore emphasize two facets of SBTG over existing approaches: learning explicitly directed couplings between neurons and learning such directed effects at varying temporal resolutions. Appendix~\ref{app:baseline_choice} discusses the rationale for the synthetic-vs-empirical baseline split, our emphasis on Spearman correlation under partial-reference ground truth and the agreement between different methods’ edge rankings.




\textbf{Multi-Lag Analysis: Cell-Type Temporal Signatures.} Having established lag-1 performance, we now ask whether SBTG's multi-lag analysis reveals biologically meaningful temporal structure. Rather than interpreting thousands of neuron-to-neuron couplings directly, we summarize the inferred edges by functional cell type---sensory ($n = 28$), interneuron ($n = 34$), and motor ($n = 11$)---and track how directed coupling strength redistributes across lags (\cref{fig:multilag}; see \cref{app:celltype-methods} in the Appendix for details). Across phases, the dominant signal concentrates at short lags, consistent with fast local computation. This also provides an internal check that SBTG is separating lag structure rather than merely reflecting shared marginal autocorrelation.

\begin{figure}[ht]
\centering
\begin{minipage}[c]{0.58\linewidth}
    \centering
    \includegraphics[width=\linewidth]{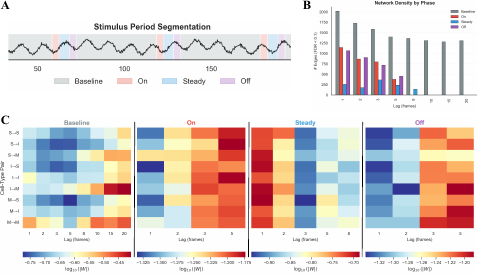}
\end{minipage}\hfill
\begin{minipage}[c]{0.40\linewidth}
    \centering
    \includegraphics[width=\linewidth]{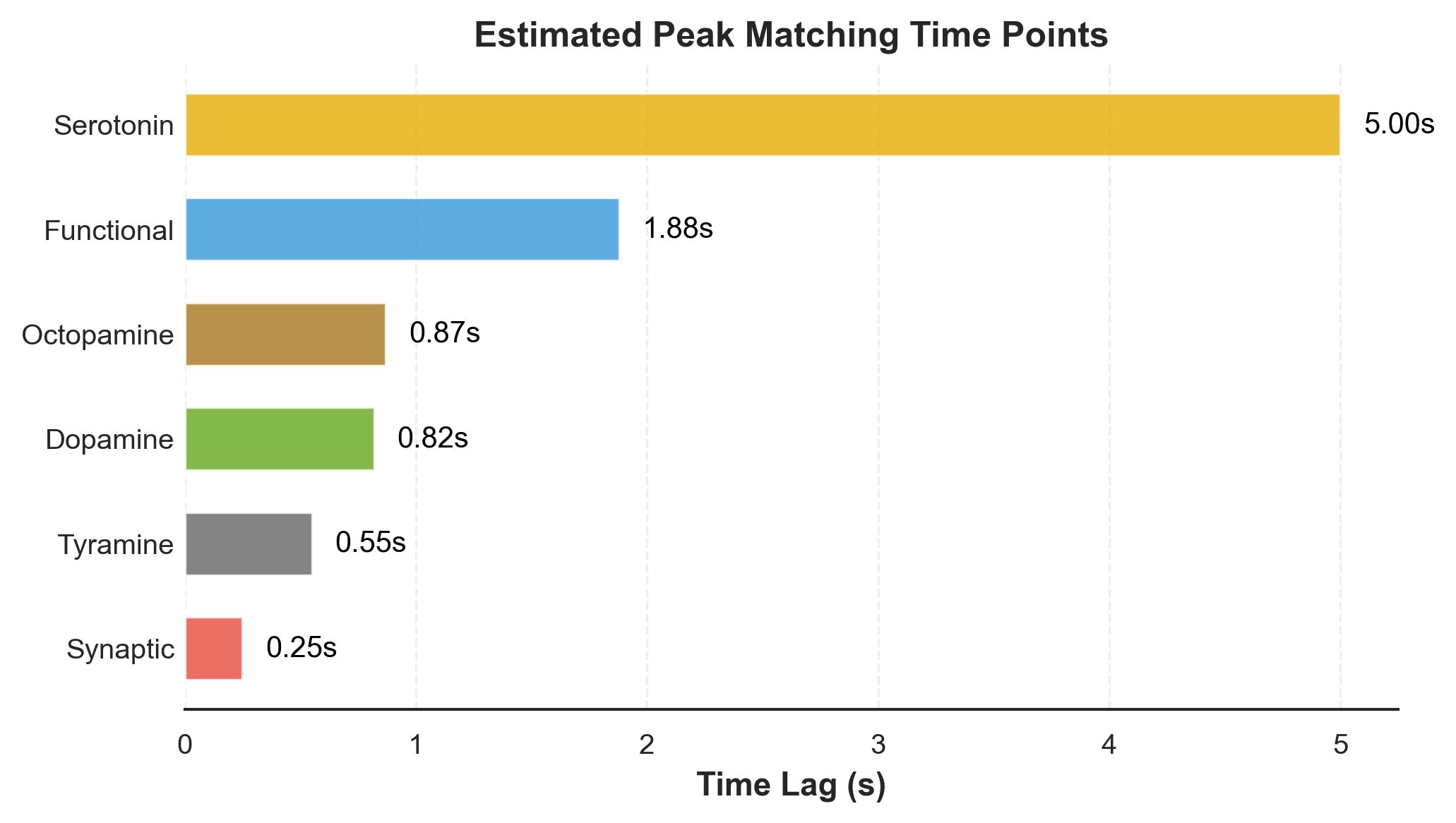}
\end{minipage}

\caption{\textbf{Stimulus Phase Analysis.} \textbf{A}: Stimulus phase segmentation used for analysis (baseline, on, steady, off).
\textbf{B}: Number of FDR-significant directed edges by lag within each phase.
\textbf{C}: Cell-type coupling strength versus lag within each phase, summarizing directed interactions between sensory (S), interneuron (I), and motor (M) groups.}
\label{fig:multilag}

\caption{\textbf{Estimated peak matching time points.}
Bar lengths indicate the interpolated time lag ($t_{\text{peak}}$) where the F1 score between the SBTG model's prediction and the respective ground-truth network is maximized.
Monoaminergic interactions (e.g., serotonin) and correlated functional activity peak at significantly longer timescales compared to the fast structural backbone.}
\label{fig:peak_lags}
\vspace{-4mm}
\end{figure}

We then refine the analysis by stratifying windows according to stimulus \emph{phase} (\cref{fig:multilag}A): Baseline (no stimulus), On (immediately after onset), Steady (sustained stimulus), and Off (immediately after removal). As expected from increasing window dimension and shrinking effective sample size, the number of FDR-significant edges decreases with lag within every phase (\cref{fig:multilag}B). However, despite On/Off being short (10 s) and therefore far more data-limited than Baseline/Steady, their inferred coupling magnitudes shift most strongly toward longer lags (\cref{fig:multilag}C). This is consistent with the hypothesis that dominant computations occur at stimulus transitions: onset and offset perturb the circuit, and their effects propagate over a longer integration horizon even when fewer windows are available \cite{kato2014temporal, ferkey2021chemosensory}.

The sustained phases exhibit the complementary signature. During stimulus presentation, coupling is strongest at short lags and attenuates as lag increases (\cref{fig:multilag}C), indicating that once a stimulus-driven regime is established, interactions concentrate at short latency rather than accumulating at long delays. Non-stimulus presentation is comparatively quiet at short lags but picks up structure at the longest lags (\cref{fig:multilag}C), consistent with baseline windows near the end of a trial that sit immediately before the next onset, so that only long-lag statistics ``reach into'' the upcoming transition.

Taken together, these results support a phase-resolved picture that is difficult to obtain from single-lag methods: onset and offset preferentially express longer-lag coupling, sustained stimulation concentrates coupling at short lags, and baseline periods show weak short-lag structure while registering long-lag changes near upcoming transitions. 


\looseness=-1
\textbf{Synaptic and Neuromodulatory Temporal Profiles.} Published biological results provide a testbed for SBTG (Figure \ref{fig:overview}E). Electrical and chemical synapses are known to traffic the fastest signals \cite{bargmann2012beyond, bargmann2013connectome}. Monoaminergic neuromodulators mediate slower state transitions with one important exception: tyramine signals a time-critical fast escape response through an ionotropic receptor (a ligand-gated ion channel) \cite{pirri2009tyramine}. Moreover, ionotropic channels, which open in direct response to ligand binding, signal faster than their metabotropic counterparts, wherein ligand binding triggers rate-limiting second-messenger systems to open separate ion channels \cite{hobert2013neuronal}. The published synaptic \cite{Cook2019}, monoamine \cite{bentley2016multilayer}, and GABA \cite{Yemini2021} connectomes permitted us to test SBTG for these biological ground truths.

For each neurotransmitter and each lag, we treat the inferred coupling magnitudes as edge scores and report AUROC/AUPRC against the corresponding reference adjacency (\cref{tab:peak_lags_appendix}). Directly inline with biological predictions, synaptic connectivity marks the fastest peak of activity, followed by tyramine as the second fastest signal, with the remaining slower acting neuromodulators peaks thereafter. Notably, dopamine suppresses octopamine signaling \cite{suo2009dopamine}, which may explain why it slightly precedes octopamine in timing. Serotonin, which in the head stimulates pharyngeal pumping in response to food \cite{avery1990effects, dag2023dissecting}, exhibited a very long lag consistent with the absence of food and pharyngeal pumping in the neural activity experiments. The Randi et al.\ functional atlas reflects its mixture of the aforementioned signaling pathways. Lastly, as expected, ionotropic GABA signaling was faster than its metabotropic counterpart (\cref{app:chem_gap}).

Notably, these temporal distinctions emerge without transmitter-specific modeling assumptions: SBTG is fit once per lag to activity windows, and transmitter identity enters only at evaluation time through an external reference graph. Taken together, our results suggest that lag-resolved score geometry can reveal meaningful timescale signatures in population dynamics, complementing the synaptic-scale (lag-1) structure emphasized by standard functional-connectivity analyses.

\section{Discussion}
\label{sec:discussion}

\looseness=-1
We studied the problem of recovering directed, lag-specific neural interactions from population recordings under nonlinear dynamics, without assuming a parametric transition model. Our central claim is that the local geometry of joint window distributions can be systematically converted into statistically testable circuit hypotheses. Score--Block Time Graphs (SBTG) instantiate this theory by transforming joint-window score geometry into signed, lag-resolved edge tests, linking high-dimensional activity trajectories to interpretable statements about \emph{who influences whom, and on what timescale}.

This pipeline is justified by a straightforward identification mechanism. For the minimal multi-block window $\z_t=(\y_t,\y_{t+1},\ldots,\y_{t+\ell})$, cross-block score products recover the lag-$\ell$ Jacobian in expectation under additive-noise dynamics, up to a noise scaling and a global sign flip. This motivates the sign-aligned coupling $\bm{C}_\ell := -\,\E[\s_{t+\ell}(\z_t)\s_t(\z_t)^\top]$. Conditioning on intermediate time points is essential: when $L\ge2$, pair-window estimates conflate lag effects through temporal dependence (an omitted-lag bias), whereas minimal multi-block windows separate timescales by construction. Cross-fitting, HAC standard errors, BY control, and null-contrast tuning make this identification practically usable in short, temporally dependent recordings.

Several limitations bound the interpretation of our results. Our sampling rate (4 Hz) cannot resolve sub-250 ms dynamics: cascades that unfold over multiple observed frames induce lag-specific dependence and are therefore recoverable by SBTG, whereas mechanisms completing within a single 250 ms bin are absorbed into the one-step sampled transition and are not separately identifiable from a concurrent pattern at this rate. We use only 189 head and 42 tail neurons of the 302 in \textit{C. elegans}, so observed couplings may include a reduced-form bias induced by hidden activity, although this is partially mitigated by relatively high coverage and cross-worm aggregation. Power decays with lag because window dimension grows and effective sample size shrinks, and our external references (Cook, Randi, and Bentley connectomes) are incomplete proxies for task- and state-dependent effective connectivity. Finally, our exact-identification result (Theorem~\ref{thm:jacobian_main}) requires additive-noise dynamics with the stationarity / moment conditions of Assumption~\ref{ass:moments}; outside that regime SBTG should be read as estimating a reduced-form effective-coupling object rather than the literal Jacobian, an interpretation we elaborate in Appendix~\ref{app:noise_assumptions}.

\looseness=-1
Within these constraints, the results support a high-level conclusion: SBTG can identify \emph{which lags matter} in a regime where data bottlenecks are severe. At lag 1, SBTG improves rank alignment with anatomical structure while producing signed, directed hypotheses. More importantly, the multi-lag, phase-stratified analysis indicates that timescale signatures are regime-dependent: onset/offset phases shift coupling toward longer lags despite being the most data-limited, while sustained stimulation concentrates coupling at short latency. Transmitter-specific evaluation further shows distinct lag profiles, with monoamines and metabotropic signaling shifting to longer-lag alignment relative to faster synaptic and ionotropic systems. Together, these findings suggest that lag-resolved score geometry can reveal multi-timescale organization in the state dynamics of neural populations.

\section*{Acknowledgments}
L.D. was supported by the European Commission's Marie Skłodowska-Curie Action Grant agreement no. 101207931.E.Y. was funded by the Esther A. \& Joseph Klingenstein Fund, the Simons Foundation, and the Hypothesis Fund.

\bibliography{refs}
\bibliographystyle{plainnat}

\appendix

\section{Reproducibility statement}
\label{app:reproducibility}

All experiments were run on a single NVIDIA RTX 5000 GPU and required between one and two total GPU-days end-to-end.

\section{Theoretical Model and Results}
\label{app:model}

\subsection{Notation}
We use the following conventions throughout:

\begin{center}
\begin{tabular}{cl}
\toprule
\textbf{Symbol} & \textbf{Meaning} \\
\midrule
$N$ & Total number of neurons in the organism \\
$n$ & Number of observed neurons in a recording ($n \leq N$) \\
$L$ & Maximum lag order in the dynamics \\
$T$ & Number of time points per recording \\
$M$ & Number of recordings (e.g., different worms) \\
$\y_t \in \R^N$ & Full population activity at time $t$ \\
$\x_t \in \R^n$ & Observed neuron activity at time $t$ \\
$O \subset [N]$ & Set of observed neuron indices \\
$\z_t$ & Set of full (Or observed) states learned for fixed lag $\ell$ \\
$H := [N] \setminus O$ & Set of hidden (unobserved) neuron indices \\
$f: \R^{LN} \to \R^N$ & Transition function (dynamics) \\
$\bm{F}_\ell \in \R^{N \times N}$ & Jacobian at lag $\ell$: $(\bm{F}_\ell)_{ji} = \partial f_j / \partial y_{t+1-\ell,i}$ \\
$\bm{\Omega} \in \R^{N \times N}$ & Noise covariance matrix \\
$\bm{s}(\bm{z}) := \nabla_{\bm{z}} \log p(\bm{z})$ & Score function of density $p$ \\
$\mu_{\ell,ji}$ & Cross-block score product: $\mathbb{E}[s_{t+\ell,j} \cdot s_{t,i}]$ \\
\bottomrule
\end{tabular}
\end{center}

\subsection{Model Definition}

\begin{assumption}[Multi-Lag Dynamics with Additive Noise]
\label{ass:dynamics_new}
The full population activity $\y_t \in \R^N$ evolves according to:
\begin{equation}
\y_{t+1} = f(\y_t, \y_{t-1}, \ldots, \y_{t-L+1}) + \bm{\eta}_t,
\label{eq:dynamics_full}
\end{equation}
where:
\begin{itemize}[leftmargin=1.5em]
\item $f: \R^{LN} \to \R^N$ is a differentiable transition function
\item $\bm{\eta}_t \overset{\text{iid}}{\sim} \mathcal{N}(0, \bm{\Omega})$ with $\bm{\Omega} \succ 0$
\item $\bm{\eta}_t$ is independent of $(\y_s)_{s \leq t}$
\end{itemize}
\end{assumption}

\begin{assumption}[Partial Observation]
\label{ass:partial_new}
For each recording $m \in [M]$, we observe $\x^{(m)}_t := \y_t[O_m]$ where $O_m \subset [N]$ is the set of observed neurons. Different recordings may observe different subsets.
\end{assumption}

\begin{assumption}[Shared Dynamics Across Recordings]
\label{ass:shared_new}
The transition function $f$ and noise covariance $\bm{\Omega}$ are identical across all recordings. This reflects species-level conservation of neural circuit structure.
\end{assumption}

\begin{assumption}[Stationarity and Moment Conditions]
\label{ass:moments}
The process $(\y_t)$ is stationary with $\E[\|\y_t\|^4] < \infty$. The stationary density $p(\y)$ is positive and smooth with $\int \|\nabla \log p(\y)\|^2 p(\y) d\y < \infty$.
\end{assumption}

\subsection{Main Theorem: Unified Identifiability}

\begin{theorem}[Unified Lag-Specific Jacobian Identifiability]
\label{thm:unified_new}
Under Assumptions \ref{ass:dynamics_new}--\ref{ass:moments}, define the window $\z_t := (\y_t, \y_{t+1}, \ldots, \y_{t+\ell})$ and the cross-block score product matrix:
$$
\bm{M}_\ell := \E\left[\s_{t+\ell}(\z_t) \, \s_t(\z_t)^\top\right] \in \R^{N \times N},
$$
where $\s_\tau(\z) := \nabla_{\y_\tau} \log p(\z)$ is the score component with respect to time $\tau$.

Then:
\begin{enumerate}[label=\textbf{(\alph*)}, leftmargin=2em]
\item \textbf{Full observation:} If all $N$ neurons are observed,
\begin{equation}
\bm{M}_\ell = -\,\bm{\Omega}^{-1} \E[\bm{F}_\ell].
\label{eq:full_id}
\end{equation}
In particular, $\supp(\bm{M}_\ell) = \supp(\E[\bm{F}_\ell])$. Moreover, if $\bm{\Omega}$ is diagonal then $\bm{C}_\ell := -\bm{M}_\ell$ satisfies $\mathrm{sign}((\bm{C}_\ell)_{ji}) = \mathrm{sign}(\E[(\bm{F}_\ell)_{ji}])$.

\item \textbf{Partial observation:} If we observe $O \subset [N]$ with $|O| = n$, define $\tilde{\bm{M}}_\ell := \E[\tilde{\s}_{t+\ell} \tilde{\s}_t^\top]$ from the marginal score on observed neurons. Then:
\begin{equation}
\tilde{\bm{M}}_\ell = -\,\bm{\Omega}_{OO}^{-1} \E[\bm{F}_\ell[O,O]] + \bm{B}_\ell,
\label{eq:partial_id_new}
\end{equation}
where $\bm{B}_\ell$ is a bias matrix depending on $\Cov(\y[H], \y[O])$. The bias vanishes if $\y[H] \perp\!\!\!\perp \y[O]$ conditional on the past.

\end{enumerate}
\end{theorem}

\subsection{Proof of Theorem~\ref{thm:unified_new}}

We first recall the main identity justifying the use of score-block matrices.

\begin{lemma}[Score-Hessian Identity]
\label{lem:score_hess}
For any density $p(\z)$ with finite Fisher information, and any coordinates $a, b$:
$$
\E[s_a(\z)\, s_b(\z)] \;=\; -\E\!\left[\frac{\partial^2 \log p(\z)}{\partial z_a\, \partial z_b}\right].
$$
\end{lemma}

\begin{proof}
By definition, $s_a(\z) = \partial_{z_a} \log p(\z) = \frac{\partial_{z_a} p(\z)}{p(\z)}$. Expanding the expectation:
\begin{align}
\E[s_a \cdot s_b]
&= \int s_a(\z)\, s_b(\z)\, p(\z)\, d\z \nonumber\\
&= \int \frac{\partial_{z_a} p(\z)}{p(\z)} \cdot s_b(\z) \cdot p(\z)\, d\z \nonumber\\
&= \int \partial_{z_a} p(\z) \cdot s_b(\z)\, d\z.
\label{eq:expand}
\end{align}
Applying integration by parts in $z_a$, with boundary terms vanishing by regularity of $p$:
\begin{align}
\int \partial_{z_a} p(\z)\, s_b(\z)\, d\z
&= -\int p(\z)\, \partial_{z_a} s_b(\z)\, d\z
= -\E[\partial_{z_a} s_b(\z)].
\label{eq:ibp}
\end{align}
Since $\partial_{z_a} s_b(\z) = \partial_{z_a}\partial_{z_b} \log p(\z)$, combining \eqref{eq:expand}--\eqref{eq:ibp} yields:
$$
\E[s_a \cdot s_b]
= -\E\!\left[\frac{\partial^2 \log p(\z)}{\partial z_a\, \partial z_b}\right]. \qedhere
$$
\end{proof}

We now prove Theorem~\ref{thm:unified_new}. Throughout, we restrict to lags $\ell \in \{1,\dots,L\}$ so that the lag-$\ell$ Jacobian $\bm{F}_\ell$ corresponds to a direct dependence of $\y_{t+\ell}$ on $\y_t$ through a single transition.

\paragraph{Part (a): full observation.}
Assume all $N$ neurons are observed. Consider the window $\z_t := (\y_t,\y_{t+1},\dots,\y_{t+\ell})$ and the joint density $p(\z_t)$. By the Markov structure induced by \eqref{eq:dynamics_full}, the joint log-density decomposes as
\begin{equation}
\log p(\z_t)
=
\log p(\y_{t+\ell} \mid \y_{t:t+\ell-1})
+
\log p(\y_{t:t+\ell-1}),
\label{eq:joint_decomp}
\end{equation}
where $\y_{t:t+\ell-1}:=(\y_t,\dots,\y_{t+\ell-1})$. The second term in \eqref{eq:joint_decomp} does not depend on $\y_{t+\ell}$, hence any mixed derivative involving $\y_{t+\ell}$ reduces to the corresponding mixed derivative of the conditional term.

By Assumption~\ref{ass:dynamics_new}, the conditional distribution is Gaussian:
$$
\y_{t+\ell} \mid \y_{t:t+\ell-1}
\sim
\mathcal{N}\!\Bigl(f(\y_{t+\ell-1},\y_{t+\ell-2},\ldots,\y_{t+\ell-L}),\,\bm{\Omega}\Bigr).
$$
Define the residual
$$
\bm{r}
:=
\y_{t+\ell} - f(\y_{t+\ell-1},\ldots,\y_{t+\ell-L}).
$$
Then the conditional log-density is
\begin{equation}
\log p(\y_{t+\ell}\mid \y_{t:t+\ell-1})
=
-\frac{1}{2}\,\bm{r}^{\top}\bm{\Omega}^{-1}\bm{r} + C,
\label{eq:cond_log}
\end{equation}
where $C$ is a normalization constant. Differentiating \eqref{eq:cond_log} with respect to $\y_{t+\ell}$ yields the (future-block) score:
\begin{equation}
\nabla_{\y_{t+\ell}} \log p(\y_{t+\ell}\mid \y_{t:t+\ell-1})
=
-\bm{\Omega}^{-1}\bm{r}.
\label{eq:future_score}
\end{equation}
Now take the mixed derivative with respect to $\y_t$. Since $\bm{r}$ depends on $\y_t$ only through the mean function $f$, we have $\nabla_{\y_t}\bm{r} = -\nabla_{\y_t} f$. Therefore,
\begin{align}
\nabla_{\y_t}\nabla_{\y_{t+\ell}}^{\top}\log p(\y_{t+\ell}\mid \y_{t:t+\ell-1})
&=
\nabla_{\y_t}\Bigl(-\bm{\Omega}^{-1}\bm{r}\Bigr)^{\top}
\nonumber\\
&=
-\nabla_{\y_t}\bigl(\bm{r}^{\top}\bigr)\bm{\Omega}^{-1}
\nonumber\\
&=
-( -\nabla_{\y_t} f^{\top})\bm{\Omega}^{-1}
=
\nabla_{\y_t} f^{\top}\bm{\Omega}^{-1}.
\label{eq:mixed_hess_intermediate}
\end{align}
Equivalently, in entrywise form, for indices $j,i \in [N]$,
\begin{align}
\frac{\partial^2}{\partial y_{t+\ell,j}\,\partial y_{t,i}}
\log p(\y_{t+\ell}\mid \y_{t:t+\ell-1})
&=
\sum_{k=1}^{N}
\frac{\partial}{\partial y_{t,i}}
\Bigl[(-\bm{\Omega}^{-1}\bm{r})_j\Bigr]
\nonumber\\
&=
\sum_{k=1}^{N}
(-\bm{\Omega}^{-1})_{jk}\,\frac{\partial r_k}{\partial y_{t,i}}
=
\sum_{k=1}^{N}
(-\bm{\Omega}^{-1})_{jk}\,\Bigl(-\frac{\partial f_k}{\partial y_{t,i}}\Bigr)
\nonumber\\
&=
(\bm{\Omega}^{-1}\bm{F}_\ell)_{j i},
\label{eq:mixed_hess_entry}
\end{align}
where $\bm{F}_\ell := \frac{\partial f}{\partial \y_t}\in \R^{N\times N}$ is the lag-$\ell$ Jacobian evaluated at the appropriate lagged state stack (and thus depends on time through the state).

Combining \eqref{eq:joint_decomp} with \eqref{eq:mixed_hess_entry}, we obtain
$$
\frac{\partial^2}{\partial y_{t+\ell,j}\,\partial y_{t,i}}
\log p(\z_t)
=
\frac{\partial^2}{\partial y_{t+\ell,j}\,\partial y_{t,i}}
\log p(\y_{t+\ell}\mid \y_{t:t+\ell-1})
=
(\bm{\Omega}^{-1}\bm{F}_\ell)_{j i}.
$$
Now apply Lemma~\ref{lem:score_hess} with coordinates $a=(t+\ell,j)$ and $b=(t,i)$:
\begin{align}
(\bm{M}_\ell)_{j i}
&=
\E\!\left[s_{t+\ell,j}(\z_t)\, s_{t,i}(\z_t)\right]
=
-\E\!\left[
\frac{\partial^2}{\partial y_{t+\ell,j}\,\partial y_{t,i}}
\log p(\z_t)
\right]
\nonumber\\
&=
-\E\!\left[(\bm{\Omega}^{-1}\bm{F}_\ell)_{j i}\right]
=
-\bigl(\bm{\Omega}^{-1}\E[\bm{F}_\ell]\bigr)_{j i}.
\label{eq:full_id_correct}
\end{align}
Hence,
\begin{equation}
\bm{M}_\ell
=
-\bm{\Omega}^{-1}\E[\bm{F}_\ell].
\label{eq:full_id_matrix}
\end{equation}
This establishes part (a), with the support and sign relations following directly from \eqref{eq:full_id_matrix}. \hfill $\square$

\paragraph{Part (b): partial observation.}
Now suppose we observe only a subset $O\subset[N]$ with $|O|=n$, and write $H:=[N]\setminus O$ for the hidden indices. Let
$$
\tilde{\z}_t := (\x_t,\x_{t+1},\ldots,\x_{t+\ell})
\qquad\text{where}\qquad
\x_\tau := \y_\tau[O],
$$
and let $\h_\tau := \y_\tau[H]$ denote the hidden components. We use the decomposition
$$
\z_t = (\tilde{\z}_t, \h_{t:t+\ell}),
\qquad
\h_{t:t+\ell} := (\h_t,\h_{t+1},\ldots,\h_{t+\ell}).
$$
Define the marginal score components
$$
\tilde{\s}_\tau(\tilde{\z}_t) := \nabla_{\x_\tau} \log p(\tilde{\z}_t),
\qquad \tau\in\{t,t+\ell\},
$$
and the corresponding cross-block score product matrix
$$
\tilde{\bm{M}}_\ell := \E\!\left[\tilde{\s}_{t+\ell}(\tilde{\z}_t)\,\tilde{\s}_t(\tilde{\z}_t)^\top\right]\in\R^{n\times n}.
$$

We first relate the marginal score to the full-data score via differentiation under the integral. Since
$$
p(\tilde{\z}_t) = \int p(\tilde{\z}_t,\h_{t:t+\ell})\, d\h_{t:t+\ell},
$$
we have, for any observed coordinate $a$ (corresponding to some $(\tau,i)$ with $i\in O$),
\begin{align}
\partial_a \log p(\tilde{\z}_t)
&=
\frac{\partial_a p(\tilde{\z}_t)}{p(\tilde{\z}_t)}
=
\frac{1}{p(\tilde{\z}_t)}
\int \partial_a p(\tilde{\z}_t,\h_{t:t+\ell})\, d\h_{t:t+\ell}
\nonumber\\
&=
\frac{1}{p(\tilde{\z}_t)}
\int \partial_a \log p(\tilde{\z}_t,\h_{t:t+\ell})\; p(\tilde{\z}_t,\h_{t:t+\ell})\, d\h_{t:t+\ell}
\nonumber\\
&=
\E\!\left[\partial_a \log p(\z_t)\mid \tilde{\z}_t\right].
\label{eq:fisher_identity_scalar}
\end{align}
Stacking these identities over $i\in O$ yields the vector form:
\begin{equation}
\tilde{\s}_\tau(\tilde{\z}_t)
=
\E\!\left[\s_{\tau,O}(\z_t)\mid \tilde{\z}_t\right],
\qquad \tau\in\{t,t+\ell\},
\label{eq:fisher_identity_vector}
\end{equation}
where $\s_{\tau,O}(\z_t):=\nabla_{\y_\tau[O]}\log p(\z_t)\in\R^{n}$ is the full-data score restricted to the observed coordinates.

Let
$$
\bm{A}:=\s_{t+\ell,O}(\z_t)\in\R^{n},
\qquad
\bm{B}:=\s_{t,O}(\z_t)\in\R^{n},
\qquad
\bm{X}:=\tilde{\z}_t.
$$
Then \eqref{eq:fisher_identity_vector} implies
$$
\tilde{\s}_{t+\ell}(\bm{X}) = \E[\bm{A}\mid \bm{X}],
\qquad
\tilde{\s}_{t}(\bm{X}) = \E[\bm{B}\mid \bm{X}].
$$
Consequently,
\begin{equation}
\tilde{\bm{M}}_\ell
=
\E\!\left[\tilde{\s}_{t+\ell}(\bm{X})\,\tilde{\s}_{t}(\bm{X})^\top\right]
=
\E\!\left[\E[\bm{A}\mid \bm{X}]\,\E[\bm{B}\mid \bm{X}]^\top\right].
\label{eq:tildeM_def_expand}
\end{equation}
We now apply the matrix-valued law of total covariance:
\begin{equation}
\E[\bm{A}\bm{B}^\top]
=
\E\!\left[\E[\bm{A}\mid \bm{X}]\,\E[\bm{B}\mid \bm{X}]^\top\right]
+
\E\!\left[\Cov(\bm{A},\bm{B}\mid \bm{X})\right].
\label{eq:total_cov}
\end{equation}
Rearranging \eqref{eq:total_cov} and using \eqref{eq:tildeM_def_expand} yields
\begin{equation}
\tilde{\bm{M}}_\ell
=
\E[\bm{A}\bm{B}^\top]
-
\E\!\left[\Cov(\bm{A},\bm{B}\mid \tilde{\z}_t)\right].
\label{eq:tildeM_totalcov}
\end{equation}
The first term is the corresponding observed-observed block of the full cross-block score matrix:
$$
\E[\bm{A}\bm{B}^\top]
=
\E\!\left[\s_{t+\ell,O}(\z_t)\,\s_{t,O}(\z_t)^\top\right]
=
(\bm{M}_\ell)_{OO}.
$$
Therefore,
\begin{equation}
\tilde{\bm{M}}_\ell
=
(\bm{M}_\ell)_{OO}
+
\bm{B}_\ell,
\qquad
\bm{B}_\ell
:=
-\E\!\left[\Cov\!\left(\s_{t+\ell,O}(\z_t),\,\s_{t,O}(\z_t)\mid \tilde{\z}_t\right)\right].
\label{eq:partial_decomp}
\end{equation}
Finally, substituting the full-observation identity \eqref{eq:full_id_matrix} into \eqref{eq:partial_decomp} gives
\begin{equation}
\tilde{\bm{M}}_\ell
=
-\bigl(\bm{\Omega}^{-1}\E[\bm{F}_\ell]\bigr)_{OO}
+
\bm{B}_\ell.
\label{eq:partial_id_correct_general}
\end{equation}
In particular, if the noise covariance is block-diagonal across observed and hidden coordinates (i.e.\ $\bm{\Omega}_{OH}=\bm{\Omega}_{HO}=\bm{0}$), then $(\bm{\Omega}^{-1})_{OO}=\bm{\Omega}_{OO}^{-1}$ and \eqref{eq:partial_id_correct_general} simplifies to
\begin{equation}
\tilde{\bm{M}}_\ell
=
-\bm{\Omega}_{OO}^{-1}\E[\bm{F}_\ell[O,O]]
+
\bm{B}_\ell.
\label{eq:partial_id_correct_blockdiag}
\end{equation}
The bias term $\bm{B}_\ell$ in \eqref{eq:partial_decomp} captures the effect of marginalizing hidden neurons and depends on the conditional variability of the full-data scores given the observed window (and thus on the coupling between hidden and observed components, e.g.\ through $\Cov(\y[H],\y[O])$). Moreover, $\bm{B}_\ell=\bm{0}$ whenever
$$
\Cov\!\left(\s_{t+\ell,O}(\z_t),\,\s_{t,O}(\z_t)\mid \tilde{\z}_t\right)=\bm{0}
$$
which holds, for example, if the hidden components are conditionally independent of the observed components given the past. This establishes part (b). \hfill $\square$

\section{Methods and Implementation}
\label{app:implementation}

This section provides documentation of our analysis pipeline, including all preprocessing steps, model configurations, hyperparameter choices, and statistical procedures. Algorithm~\ref{alg:sbtg} summarizes the end-to-end SBTG procedure; the subsections below describe each step in detail.

\begin{algorithm}[h!]
\caption{Score--Block Time Graphs (SBTG)}
\label{alg:sbtg}
\begin{center}
\begin{minipage}{0.85\textwidth}
\begin{algorithmic}[1]
\STATE \textbf{Input:} Multi-animal recordings $\{X^{(u)}\}_{u=1}^M$; lag set $\mathcal{L}$; FDR level $\alpha$
\STATE \textbf{Output:} Lag-specific signed adjacencies $\{A^{(\ell)}\}_{\ell \in \mathcal{L}}$
\FOR{each lag $\ell \in \mathcal{L}$}
    \STATE Build minimal $(\ell{+}1)$-block windows; pool across worms
    \STATE Select hyperparameters via null contrast (Eq.~\ref{eq:null_contrast})
    \STATE 5-fold cross-fit: train score model, evaluate on held-out
    \FOR{each ordered neuron pair $(i, j)$}
        \STATE Compute product series $Y_t = \widehat\s_{t+\ell,j} \cdot \widehat\s_{t,i}$
        \STATE Estimate $\widehat\mu_{\ell,ji}$ and HAC standard error
    \ENDFOR
    \STATE Aggregate across worms via inverse-variance meta-analysis
    \STATE Apply BY FDR control at level $\alpha$; store $A^{(\ell)}_{ji} = \mathrm{sign}(\widehat\mu_{\ell,ji})$
\ENDFOR
\STATE \textbf{return} $\{A^{(\ell)}\}_{\ell \in \mathcal{L}}$
\end{algorithmic}
\end{minipage}
\end{center}
\end{algorithm}

\subsection{Data Preprocessing and Quality Control}
\label{app:preprocessing}

\subsubsection{Raw Data Sources}

Our analysis integrates three primary data sources. The NeuroPAL calcium imaging data~\citep{Yemini2021} consists of head and tail activity recordings sampled at 4 Hz (250 ms per frame) over 240-second recordings. Three chemical stimuli--- the attractive odors 2-butanone and 2,3-pentanedione, and an aversive quantity of NaCl (160mM) ---were used. The structural connectome from \citet{white1986structure, Cook2019} provides anatomical ground truth from electron microscopy reconstruction, including chemical synapses and gap junctions which we use as binary ground truth for edge prediction. Molecular expression data using NeuroPAL and neural identification reporters adds ground truth for the ionotropic GABAa connectome and metabotropic GABAb connectome \cite{Yemini2021}. The Bentley connectome \cite{bentley2016multilayer} provides ground truth from molecular expression reporters for monoaminergic (dopamine, octopamine, serotonin, tyramine) neuromodulatory edge prediction. The Randi et al.\ functional atlas~\citep{Randi2023} offers functional ground truth from optogenetic impulse response experiments, quantifying a second dimension of functional connectivity.

\subsubsection{Neuron Identification and Alignment}

\paragraph{Coverage Statistics.}
\begin{table}[h!]
\centering
\caption{Neuron coverage across recordings.}
\label{tab:coverage}
\begin{tabular}{lrr}
\toprule
\textbf{Category} & \textbf{Count} & \textbf{Notes} \\
\midrule
Total neurons (\emph{C. elegans}) & 302 & Hermaphrodite \\
Head neurons identified & $189$ & NeuroPAL \\
Tail neurons identified & $42$ & NeuroPAL \\
Neurons in $\geq 15$ worms & 80 & Our analysis set \\
Median appearances per neuron & 18 & Range: 15--20 \\
Worms with $\geq 70$ neurons & 21 & Post-imputation \\
\bottomrule
\end{tabular}
\end{table}

We require neurons to appear in $\geq 15$ worms, balancing coverage breadth with per-neuron sample size.

\paragraph{D/V Subtype Collapsing.}
NeuroPAL identifies dorsal/ventral subtypes (e.g., RMDD, RMDV) not distinguished in Cook's cell-class connectome. We collapse these to enable connectome alignment, recovering approximately 20 neurons that would otherwise be excluded. Table~\ref{tab:dv_collapse} lists all collapsed pairs.

\begin{table}[h!]
\centering
\caption{Dorsal/ventral subtype collapsing.}
\label{tab:dv_collapse}
\begin{small}
\begin{tabular}{ll}
\toprule
\textbf{Subtypes} & \textbf{Parent Class} \\
\midrule
RMDD, RMDV & RMD \\
SMDD, SMDV & SMD \\
RMEV, RMED & RME \\
RMDL, RMDR & RMD (L/R) \\
\bottomrule
\end{tabular}
\end{small}
\end{table}

\subsection{Hyperparameter Optimization}
\label{app:hp_optimization}

\subsubsection{Null Contrast Objective}

\paragraph{Motivation.}
Denoising score matching (DSM) loss measures how well $\widehat{\bm{s}}_\theta$ approximates the true score, but this does not guarantee edge recovery. In validation experiments, DSM validation loss showed \emph{negative} correlation with biological AUROC ($r = -0.15$ for Cook). This occurs because lower DSM loss can be achieved by fitting noise---learning high-frequency structure that does not correspond to real circuit connectivity.

\paragraph{Definition.}
The null contrast objective compares real signal strength to a null distribution:
\begin{equation}
\mathrm{NC} = \frac{\mathrm{mean}(|\widehat{\mu}_{\ell,ji}|, \, j \neq i)}{\mathrm{mean}(|\widehat{\mu}^{\mathrm{null}}_{\ell,ji}|, \, j \neq i)},
\label{eq:nc_detailed}
\end{equation}
where $\widehat{\mu}_{\ell,ji} = \frac{1}{T_{\mathrm{win}}} \sum_{t=1}^{T_{\mathrm{win}}} \widehat{s}_{t+\ell,j}(z_t) \cdot \widehat{s}_{t,i}(z_t)$ is the cross-block score product,
 $\widehat{\mu}^{\mathrm{null}}$ is computed after randomly permuting the lag-$\ell$ block scores across windows, and higher NC indicates stronger structured signal relative to permutation baseline. Permutation preserves marginal statistics (each block's score distribution) while breaking the joint structure. NC $> 1$ indicates the real data exhibits stronger cross-lag coupling than would occur by chance.

\paragraph{Implementation via Optuna.}
We use Optuna's TPE (Tree-structured Parzen Estimator) sampler for hyperparameter search \citep{akiba2019optuna}:

\begin{algorithm}[h!]
\caption{Hyperparameter Optimization via Null Contrast}
\label{alg:hp_optuna}
\begin{center}
\begin{minipage}{0.85\textwidth}
\begin{algorithmic}[1]
\scriptsize
\STATE \textbf{Input:} Training windows $\{z_t\}_{t=1}^N$, number of trials $K$
\STATE \textbf{Output:} Optimal hyperparameters $\theta^*$
\STATE Initialize Optuna study with TPE sampler
\FOR{trial $k = 1, \ldots, K$}
    \STATE Sample hyperparameters: $\sigma_{\mathrm{noise}}$, hidden\_dim, num\_layers, lr
    \STATE Split data into train/validation (80/20)
    \STATE Train score model on train set (30 epochs for speed)
    \STATE Compute scores on validation set
    \STATE Compute $\widehat{\mu}$ from validation scores
    \STATE Permute lag-$\ell$ scores, compute $\widehat{\mu}^{\mathrm{null}}$
    \STATE Compute NC = mean($|\widehat{\mu}|$) / mean($|\widehat{\mu}^{\mathrm{null}}|$)
    \STATE Report NC to Optuna
\ENDFOR
\STATE \textbf{return} Hyperparameters with highest NC
\end{algorithmic}
\end{minipage}
\end{center}
\end{algorithm}

\begin{table}[h!]
\centering
\caption{Hyperparameter search space for Optuna optimization.}
\label{tab:hp_search_space}
\begin{tabular}{lll}
\toprule
\textbf{Parameter} & \textbf{Range} & \textbf{Distribution} \\
\midrule
$\sigma_{\mathrm{noise}}$ (DSM noise) & [0.01, 0.30] & Continuous \\
hidden\_dim & \{32, 64, 128\} & Categorical \\
num\_layers & \{2, 3\} & Categorical \\
learning rate & $[10^{-4}, 10^{-2}]$ & Log-uniform \\
\midrule
\multicolumn{3}{l}{\textit{Fixed across trials:}} \\
batch\_size & 128 & - \\
train\_frac & 0.8 & - \\
epochs (tuning) & 30 & Reduced for speed \\
epochs (final) & 100 & Full training \\
\bottomrule
\end{tabular}
\end{table}

\paragraph{Number of Trials.}
We use $K = 20$ trials for synthetic benchmarks and $K = 50$ for empirical analysis. Validation experiments showed convergence by 20 trials, with diminishing returns beyond 50.


\subsection{Statistical Inference}
\label{app:inference}

\subsubsection{Cross-Fitting Protocol}

To avoid overfitting bias, all edge statistics are computed on held-out data:

\begin{algorithm}[h!]
\caption{5-Fold Cross-Fitting}
\label{alg:crossfit}
\begin{center}
\begin{minipage}{0.85\textwidth}
\begin{algorithmic}[1]
\scriptsize
\STATE \textbf{Input:} Windows $\{z_t^{(u)}\}$ with stimulus IDs $\{u\}$
\STATE Assign each unique stimulus $u$ to one of 5 folds
\FOR{fold $f = 1, \ldots, 5$}
    \STATE $\mathcal{D}_{\mathrm{train}} \gets$ windows from stimuli not in fold $f$
    \STATE $\mathcal{D}_{\mathrm{test}} \gets$ windows from stimuli in fold $f$
    \STATE Train score model $\widehat{\bm{s}}_\theta$ on $\mathcal{D}_{\mathrm{train}}$
    \STATE Compute scores on $\mathcal{D}_{\mathrm{test}}$: $\{\widehat{\bm{s}}_t\}_{t \in \mathcal{D}_{\mathrm{test}}}$
\ENDFOR
\STATE Concatenate all held-out scores
\STATE Compute edge statistics from held-out scores only
\end{algorithmic}
\end{minipage}
\end{center}
\end{algorithm}

\paragraph{Remark.} Folds are stratified by stimulus sequence (worm $\times$ trial), not by individual windows. This prevents data leakage from temporally adjacent windows and ensures independence between train and test sets.

\subsubsection{HAC Standard Errors}

\paragraph{Problem.}
The score product series $Y_t = \widehat{s}_{t+\ell,j} \cdot \widehat{s}_{t,i}$ exhibits strong autocorrelation due to overlapping windows. For lag-1 with 2-block windows, adjacent observations share the entire first block. Standard $t$-tests assuming i.i.d. observations would dramatically overstate significance.

\paragraph{Solution: Newey--West HAC Estimator.}
We use the Newey--West heteroskedasticity and autocorrelation consistent (HAC) variance estimator \citep{NeweyWest1987}:
\begin{equation}
\widehat{\mathrm{Var}}_{\mathrm{HAC}}(\bar{Y}) = \frac{1}{N}\left(\widehat{\gamma}_0 + 2\sum_{h=1}^m \left(1 - \frac{h}{m+1}\right) \widehat{\gamma}_h\right),
\label{eq:hac_var}
\end{equation}
where $\widehat{\gamma}_h = \frac{1}{N} \sum_{t=1}^{N-h} (Y_t - \bar{Y})(Y_{t+h} - \bar{Y})$ is the lag-$h$ autocovariance, $m$ is the bandwidth parameter (number of lags to include), and the kernel $\left(1 - \frac{h}{m+1}\right)$ ensures positive semi-definiteness.

\paragraph{Bandwidth Selection.}
We use $m = 7$ based on typical autocorrelation decay in our 4 Hz neural data, the rule-of-thumb $m \approx N^{1/4}$, and sensitivity analysis showing stable results for $m \in [5, 10]$.

The $t$-statistic is then:
\begin{equation}
t_{ji} = \frac{\sqrt{N} \, \widehat{\mu}_{\ell,ji}}{\sqrt{\widehat{\mathrm{Var}}_{\mathrm{HAC}}(Y_{ji})}},
\end{equation}
which we convert to two-sided $p$-values via the standard normal approximation.

\subsubsection{FDR Control}

\paragraph{Multiple Testing Burden.}
With $n = 80$ neurons, we test $n(n-1) = 6320$ directed edges per lag. At $\alpha = 0.05$ with independent tests, we would expect $\sim 316$ false discoveries. At lag-5, we test $5 \times 6320 = 31600$ edges.

\paragraph{BY Procedure.}
We use the BY procedure \citep{BenjaminiYekutieli2001}, which controls FDR under \emph{arbitrary dependence} among $p$-values:
\begin{equation}
\text{Reject } H_{(i)} \text{ if } p_{(i)} \leq \frac{i}{m \cdot c(m)} \alpha,
\end{equation}
where $p_{(1)} \leq \cdots \leq p_{(m)}$ are ordered $p$-values, $m = 6320$, and $c(m) = \sum_{j=1}^m \frac{1}{j} \approx \log(m) + \gamma$.

\paragraph{Choice of $\alpha$.}
We use $\alpha = 0.10$ as our FDR level, slightly more liberal than the conventional 0.05, reflecting the exploratory analysis context (hypothesis generation), validation against independent connectomes (false positives will not align), and typical neuroscience practice for functional connectivity. Sensitivity analysis in \cref{app:sensitivity-analysis} examines $\alpha \in \{0.01, 0.05, 0.10, 0.20\}$ for BH and BY FDR control methods.

\begin{figure}[ht]
    \centering
    \includegraphics[width=0.6\linewidth]{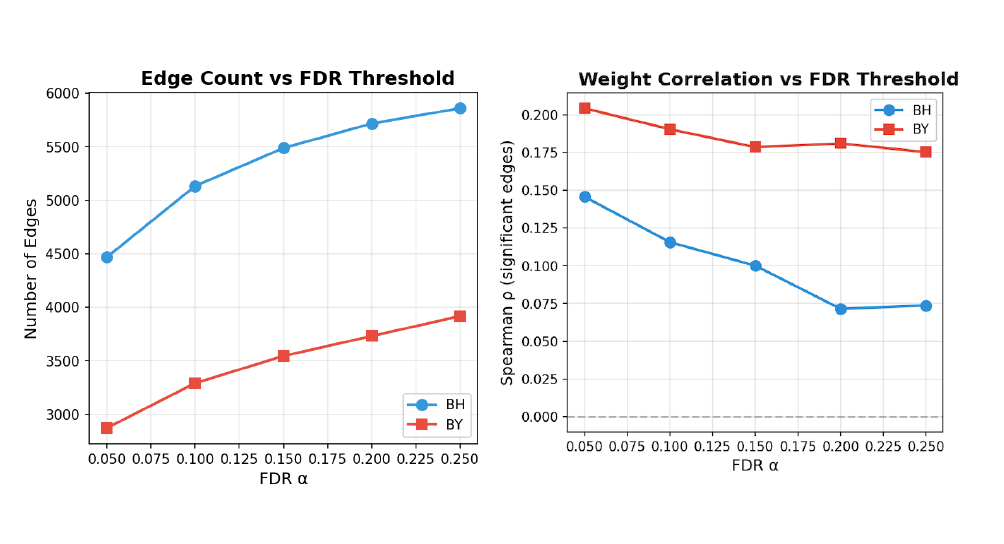}
    \caption{Left: number of discovered directed edges as a function of the nominal FDR level $\alpha$ under Benjamini--Hochberg (BH) and Benjamini--Yekutieli (BY). Right: Spearman rank correlation with an external structural benchmark, computed over the edges passing the corresponding threshold, versus $\alpha$. BY is more conservative, reflecting its validity under arbitrary dependence among test statistics.}
    \label{app:sensitivity-analysis}
\end{figure}

\subsection{Remarks}

\paragraph{Population object vs.\ practical estimator.}
\label{app:pop_vs_est}
Theorem~\ref{thm:unified_new} characterizes the \emph{population object} recovered by the exact joint-window score: under additive-noise dynamics, the identity $\bm{M}_\ell = -\bm{\Omega}^{-1}\E[\bm{F}_\ell]$ holds in expectation regardless of how the score is parameterized. The structured score model in Eq.~\eqref{eq:score_model} is a practical \emph{estimator} chosen to make this population object recoverable from short, temporally dependent recordings; it is not a claim that this specific functional form is the unique or universal score family. The structural choices in Eq.~\eqref{eq:score_model} (cross-block coupling matrices $W_r$ separated from within-block MLPs $g_k$) regularize the estimator toward the cross-block geometry that the theorem identifies, but they do not expand what the theorem proves: with an unrestricted (and consistent) score estimator the same identity would still hold in the population. The role of the structured form is to make the same target estimable in a data-limited regime.

\paragraph{Minimal multi-block windows vs.\ time-unrolled inputs.}
\label{app:multiblock_vs_unroll}
Naive time-unrolling --- concatenating lagged inputs $(\x_{t-1}, \x_{t-2}, \ldots)$ as features for a single model --- is not equivalent to SBTG's minimal multi-block window. With $L \ge 2$, pair-window estimates of the form $\E[\s(\x_{t+\ell})\,\s(\x_t)^\top]$ conflate the direct lag-$\ell$ effect with indirect pathways through intermediate time points $\x_{t+1},\dots,\x_{t+\ell-1}$ --- a temporal omitted-variable bias. SBTG's minimal multi-block window $\z^{(\ell)}_t = (\x_t, \x_{t+1}, \ldots, \x_{t+\ell})$ explicitly conditions on those intermediate states, so the cross-block score product $-\E[\s_{t+\ell}\s_t^\top]$ recovers the direct lag-$\ell$ Jacobian by construction. Time-unrolling alone does not provide this conditioning, and methods that simply concatenate lagged copies of the data inherit the omitted-lag bias even when they nominally accept long histories.

\paragraph{Sensitivity to additive-noise and stationarity assumptions.}
\label{app:noise_assumptions}
The exact identification result of Theorem~\ref{thm:jacobian_main} requires additive-noise dynamics, $\bm{x}_{t+1} = f(\bm{x}_t, \ldots, \bm{x}_{t-L+1}) + \bm{\varepsilon}_t$, with the stationarity / moment conditions of Assumption~\ref{ass:moments}. Real neural data are unlikely to satisfy these exactly. Our intended claim, however, is that SBTG targets a lag-specific coupling object directly from the joint-window distribution without requiring a parametric transition model. Outside the idealized theorem setting we view SBTG as estimating a reduced-form but still interpretable effective-coupling object --- the closest population analogue of the lag-$\ell$ Jacobian under whatever distribution actually generates the data. In high-dimensional, partially observed, non-stationary regimes characteristic of neuroscience recordings, this is a strength rather than a weakness: SBTG retains a probability-distribution-based estimand that does not depend on architectural and training choices the way RNN- or attention-based connectivity proxies (e.g.,~LINT, NetFormer) do.

\subsection{Baseline Methods Implementation}
\label{app:baselines}

We compare SBTG against seven baseline methods spanning linear/nonlinear and optimization-based approaches.

\begin{table}[h!]
\centering
\caption{Baseline method specifications.}
\label{tab:baseline_specs}
\begin{small}
\begin{tabular}{@{}l p{0.72\textwidth}@{}}
\toprule
\textbf{Method} & \textbf{Implementation Details} \\
\midrule
\textbf{Pearson} & Lag-$\ell$ cross-correlation: $\mathrm{corr}(x_{t+\ell,j}, x_{t,i})$. Pooled across worms. No regularization. \\
\midrule
\textbf{Partial Correlation} & Via precision matrix $\Omeg = \Sig^{-1}$. $\Omeg_{ji}$ estimates $\partial x_j / \partial x_i$ conditional on all others. Uses GraphicalLassoCV for regularization. \\
\midrule
\textbf{Graphical LASSO} & $\widehat{\Omeg} = \arg\min_{\Omeg} \left\{-\log\det(\Omeg) + \mathrm{tr}(\Omeg\widehat{\Sig}) + \lambda\|\Omeg\|_1\right\}$. Cross-validated $\lambda$ via scikit-learn. \\
\midrule
\textbf{Granger Causality} & Per-pair $F$-test: Does $x_{t-1:t-\ell,i}$ improve prediction of $x_{t,j}$ beyond own lags? Fit per-worm, average $F$-statistics. Implementation: \texttt{statsmodels.tsa.stattools.grangercausalitytests}. \\
\midrule
\textbf{VAR-LASSO} & Fit VAR($\ell$) via $\ell_1$-penalized regression. Per-worm fitting, average coefficients. Regularization: $\alpha = 0.1$ \\
\midrule
\textbf{VAR-Ridge} & $\ell_2$-penalized VAR. $\alpha = 1.0$. More stable than LASSO for correlated predictors. \\
\midrule
\textbf{VAR-LiNGAM} & Causal discovery via independent components \citep{Hyvarinen2013lingam}. Assumes linear non-Gaussian additive noise. Implementation: \texttt{lingam} package. \\
\bottomrule
\end{tabular}
\end{small}
\end{table}

\paragraph{Remark.}
For VAR and Granger baselines, we fit models \emph{per-worm} and average the resulting parameters/statistics. This avoids spurious temporal dependencies at worm boundaries that would arise from concatenation.

\subsubsection{Deep-Learning Baselines}
\label{app:dl_baselines}

Beyond the classical and causal-discovery baselines above, we compare against three recent deep-learning architectures for connectivity inference. All three are trained on identical sliding-window calcium-imaging data and evaluated against the Cook structural connectome at lag~1.

\begin{itemize}[leftmargin=1.5em, itemsep=2pt]
  \item \textbf{NRI}~\citep{Kipf2018NRI}: a variational autoencoder for interacting systems.
  \item \textbf{NetFormer}~\citep{Lu2025NetFormer}: a transformer-inspired next-step predictor for neural dynamics. 
  \item \textbf{LINT}~\citep{Valente2022LINT}: a method that fits low-rank rate RNNs to neural trajectories.
\end{itemize}

\begin{table}[h!]
\centering
\caption{\textbf{Deep-learning baselines on the Cook structural connectome (lag~1, $n{=}80$).} SBTG outperforms all three deep-learning baselines on every metric. NRI, NetFormer, and LINT remain near chance on AUROC, indicating that their inductive biases (latent-variable VAE, attention-based prediction, low-rank dynamical reconstruction) do not naturally produce lag-specific directed edge tests in this regime; SBTG's cross-block score geometry, by contrast, is designed to recover them by construction.}
\label{tab:dl_baselines_cook}
\begin{small}
\begin{tabular}{lcccc}
\toprule
Method & AUROC & AUPRC & F1 & Correlation \\
\midrule
\textbf{SBTG (lag 1)} & \textbf{0.581} & \textbf{0.289} & \textbf{0.284} & \textbf{0.202} \\
NRI~\citep{Kipf2018NRI} & 0.507 & 0.208 & 0.173 & 0.033 \\
NetFormer~\citep{Lu2025NetFormer} & 0.505 & 0.205 & 0.176 & 0.008 \\
LINT~\citep{Valente2022LINT} & 0.503 & 0.204 & 0.163 & 0.006 \\
\bottomrule
\end{tabular}
\end{small}
\end{table}

\subsubsection{Synthetic vs.\ Empirical Baseline Choice and Metric Rationale}
\label{app:baseline_choice}

The synthetic and empirical experiments use deliberately different baseline sets, and the empirical evaluation reports both AUROC and Spearman rank correlation. We explain both choices here.

\paragraph{Why the baseline sets differ.}
The synthetic suite (Tables~\ref{tab:synthetic}, \ref{tab:res_var}, \ref{tab:res_poisson}, \ref{tab:res_hawkes}) tests \emph{exact recovery} against fully specified data-generating processes, so we benchmark against methods that target structural recovery under known dynamical assumptions: causal-discovery (DYNOTEARS, PCMCI$^+$, VAR-LiNGAM) and regularized regression (VAR-LASSO, VAR-Ridge). The empirical \emph{C.~elegans} setting (Table~\ref{tab:lag1_baselines}) evaluates \emph{rank alignment with partial external references} (Cook structural connectome, Randi et al.\ functional atlas), where ground truth is incomplete and noisy; here we benchmark against well-understood correlation/regression baselines (Pearson, cross-correlation, Granger, Glasso) plus the deep-learning architectures of Section~\ref{app:dl_baselines} (NRI, NetFormer, LINT). Both baseline sets are domain-appropriate, but they answer different questions. The synthetic comparators are designed for sharp ground truth; the empirical comparators are designed for the partial-observation, partial-reference regime of population recordings.

\paragraph{Why Spearman correlation alongside AUROC.}
\label{app:spearman_rationale}
On the empirical Cook benchmark we report both AUROC and Spearman rank correlation $\rho$. The two metrics answer different questions in this regime. AUROC treats every unannotated edge as equally negative, including pairs that may be biologically plausible but absent from the structural reference (the Cook connectome is itself partial, and absence of an edge does not imply absence of a functional interaction). Spearman correlation, in contrast, asks whether the inferred edges that SBTG ranks more strongly align with the biologically supported relationships in the reference, irrespective of where any binary cutoff falls. Our intended real-data claim is therefore not that SBTG solves binary connectome recovery from activity alone, but that it recovers a more accurate \emph{relative ordering} of directed interactions and their multi-timescale organization. We retain AUROC for comparison with prior literature, while emphasizing $\rho$ as the more interpretable measure under partial-reference ground truth.

\paragraph{Agreement across methods.}
On the 80-neuron recordings at lag 1, Spearman correlations between edge rankings were $0.17$ for SBTG versus Pearson, $0.10$ for SBTG versus Granger, and $0.24$ for Pearson versus Granger. Among their respective top $100$ edges, SBTG and Pearson shared $46$ edges, while SBTG and Granger shared $8$. These results show that similar benchmark performance accompanied different individual learned-edge rankings.

\subsection{Synthetic Benchmark Suite}
\label{app:synthetic}

To validate SBTG's ability to recover known structure, we implemented a comprehensive synthetic benchmark suite.

\subsubsection{Data Generating Processes}

We test four synthetic families, each with lag-1 and lag-2 ground truth:

\begin{table}[h!]
\centering
\caption{Synthetic data families.}
\label{tab:synthetic_families}
\begin{small}
\begin{tabular}{@{}l p{0.78\textwidth}@{}}
\toprule
\textbf{Family} & \textbf{Dynamics} \\
\midrule
\textbf{VAR(2)} & $x_{t+1} = A_1 x_t + A_2 x_{t-1} + \varepsilon_t$, $\varepsilon_t \sim \mathcal{N}(0, \sigma^2 I)$. Ground truth: $\supp(A_1)$, $\supp(A_2)$. \\
\midrule
\textbf{Poisson GLM} & $\lambda_t = \exp(\alpha + A_1 x_t + A_2 x_{t-1} + s_t)$, $x_{t+1} \sim \mathrm{Poisson}(\lambda_t)$. Count data with stimulus drive $s_t$. \\
\midrule
\textbf{Hawkes-like} & $\lambda_t = \mathrm{softplus}(\alpha + A_1 x_t + A_2 x_{t-1})$, $x_{t+1} \sim \mathrm{Poisson}(\lambda_t)$. Self-exciting point process. \\
\midrule
\textbf{Tanh VAR(2)} & $x_{t+1} = \tanh(W_1 x_t + W_2 x_{t-1}) + \varepsilon_t$. Saturating nonlinearity like neural firing rates. \\
\bottomrule
\end{tabular}
\end{small}
\end{table}

\paragraph{Sparsity and Scaling.}
All adjacency matrices use 10\% sparsity (6--7 edges per neuron on average for $n=10$), lag-1 scale of 0.5--0.8 (stronger effects), lag-2 scale of 0.25--0.4 (weaker delayed effects), and spectral scaling for VAR families to ensure stability.

\subsubsection{Experimental Design}

\begin{table}[h!]
\centering
\caption{Synthetic benchmark configuration.}
\label{tab:synthetic_config}
\begin{tabular}{lll}
\toprule
\textbf{Factor} & \textbf{Levels} & \textbf{Values} \\
\midrule
Data family & 4 & VAR, Poisson, Hawkes, Tanh \\
Noise level & 2 & Low ($\sigma=0.1$), High ($\sigma=0.5$) \\
Sequence length & 2 & Short (300), Long (800) \\
Random seed & 2 & 0, 1 \\
\midrule
\multicolumn{3}{l}{\textbf{Total combinations: $4 \times 2 \times 2 \times 2 = 32$ datasets}} \\
\bottomrule
\end{tabular}
\end{table}

\paragraph{Network Size.}
We use $n = 10$ neurons and $m = 3$ independent stimulus sequences per dataset. These dimensions are smaller than biological scale ($n=80$) to enable fast iteration and comprehensive method comparison.

\paragraph{Methods Evaluated}

We run SBTG with 20 Optuna trials each using null contrast objective, testing two statistical configurations (HAC bandwidth $m \in \{5, 7\}$) with BY FDR at $\alpha = 0.10$, yielding 6 fits per dataset. We compare against seven baseline methods: VAR-LASSO ($\alpha = 0.1$), VAR-Ridge ($\alpha = 1.0$), VAR-LiNGAM~\citep{Hyvarinen2013lingam}, Poisson-GLM (for count families), PCMCI+~\citep{Runge2020PCMCIplus} and DYNOTEARS~\citep{pamfil2020dynotearsstructurelearningtimeseries}. All baselines use per-worm fitting with coefficient averaging (VAR methods) or per-worm F-test averaging (Granger).

\paragraph{Evaluation Metrics}

For each method and dataset, we compute binary classification metrics for lag-1 (AUROC, AUPRC, F1, Precision, Recall), multi-lag AUROC (separate values for lag-1 and lag-2 predictions), and Spearman correlation (rank correlation between predicted weights and ground truth).

\subsection{Full Synthetic Results}
\label{app:synthetic_full}

We present the full benchmark results for Linear VAR, Poisson, and Hawkes datasets below. The Nonlinear Tanh results are presented in the main text (\cref{tab:synthetic}).

\begin{table}[h!]
\centering
\caption{\textbf{Linear VAR Results.} SBTG outperforms linear/nonlinear baselines.}
\label{tab:res_var}
\begin{small}
\begin{tabular}{lccc}
\toprule
Method & F1 Score & AUROC & AUPRC \\
\midrule
\textbf{SBTG} & \textbf{0.39 $\pm$ 0.20} & \textbf{0.72 $\pm$ 0.16} & \textbf{0.25 $\pm$ 0.15} \\
DYNOTEARS & 0.14 $\pm$ 0.16 & 0.56 $\pm$ 0.06 & 0.19 $\pm$ 0.11 \\
VAR-LASSO & 0.21 $\pm$ 0.04 & 0.59 $\pm$ 0.04 & 0.17 $\pm$ 0.03 \\
PCMCI+ & 0.13 $\pm$ 0.09 & 0.52 $\pm$ 0.09 & 0.10 $\pm$ 0.00 \\
VAR-Ridge & 0.17 $\pm$ 0.03 & 0.46 $\pm$ 0.14 & 0.13 $\pm$ 0.06 \\
VAR-LiNGAM & 0.12 $\pm$ 0.10 & 0.52 $\pm$ 0.05 & 0.10 $\pm$ 0.03 \\
\bottomrule
\end{tabular}
\end{small}
\end{table}

\begin{table}[h!]
\centering
\caption{\textbf{Poisson GLM Results.} SBTG handles count data effectively.}
\label{tab:res_poisson}
\begin{small}
\begin{tabular}{lccc}
\toprule
Method & F1 Score & AUROC & AUPRC \\
\midrule
\textbf{SBTG} & \textbf{0.25 $\pm$ 0.12} & \textbf{0.66 $\pm$ 0.16} & {0.15 $\pm$ 0.06} \\
DYNOTEARS & 0.03 $\pm$ 0.05 & 0.48 $\pm$ 0.04 & 0.09 $\pm$ 0.01 \\
VAR-LASSO & 0.23 $\pm$ 0.12 & 0.58 $\pm$ 0.08 & \textbf{0.16 $\pm$ 0.06} \\
PCMCI+ & 0.19 $\pm$ 0.07 & 0.56 $\pm$ 0.07 & 0.11 $\pm$ 0.02 \\
VAR-Ridge & 0.16 $\pm$ 0.02 & 0.53 $\pm$ 0.09 & 0.15 $\pm$ 0.08 \\
VAR-LiNGAM & 0.09 $\pm$ 0.14 & 0.51 $\pm$ 0.05 & 0.12 $\pm$ 0.06 \\
\bottomrule
\end{tabular}
\end{small}
\end{table}

\begin{table}[h!]
\centering
\caption{\textbf{Hawkes Process Results.} Challenging point-process dynamics.}
\label{tab:res_hawkes}
\begin{small}
\begin{tabular}{lccc}
\toprule
Method & F1 Score & AUROC & AUPRC \\
\midrule
\textbf{SBTG} & \textbf{0.20 $\pm$ 0.09} & \textbf{0.58 $\pm$ 0.15} & {0.12 $\pm$ 0.05} \\
DYNOTEARS & 0.03 $\pm$ 0.08 & 0.50 $\pm$ 0.03 & 0.10 $\pm$ 0.02 \\
VAR-LASSO & 0.17 $\pm$ 0.09 & 0.55 $\pm$ 0.06 & \textbf{0.14 $\pm$ 0.04} \\
PCMCI+ & 0.14 $\pm$ 0.05 & 0.51 $\pm$ 0.05 & 0.09 $\pm$ 0.01 \\
VAR-Ridge & 0.16 $\pm$ 0.02 & 0.51 $\pm$ 0.07 & \textbf{0.14 $\pm$ 0.04} \\
VAR-LiNGAM & 0.03 $\pm$ 0.08 & 0.50 $\pm$ 0.03 & 0.10 $\pm$ 0.02 \\
\bottomrule
\end{tabular}
\end{small}
\end{table}

\begin{table}[h!]
\centering
\caption{\textbf{Nonlinear Tanh at empirical scale ($n{=}80$).} Extension of Table~\ref{tab:synthetic} from the $n{=}10$ to $n{=}80$. We evaluate the four $(\text{noise}, \text{length})$ variants used elsewhere in the synthetic suite: noise $\in \{\text{low}, \text{high}\}$ and length $\in \{\text{short } (T{=}300), \text{long } (T{=}3000)\}$. SBTG remains above chance on the two high-noise variants, achieving per-cell AUROC of $0.68$ at high/short and $0.82$ at high/long; the gain with $T_{\text{long}}{=}3000$ vs $T_{\text{short}}{=}300$ is consistent with our windows-per-edge analysis.  The aggregated rows below report mean $\pm$ std across all four variants, including the failed low-noise cells. Classical and causal-discovery baselines remain at chance throughout.}
\label{tab:res_tanh_n80}
\begin{small}
\begin{tabular}{lccc}
\toprule
Method & F1 Score & AUROC & AUPRC \\
\midrule
\textbf{SBTG} & \textbf{0.35 $\pm$ 0.18} & \textbf{0.62 $\pm$ 0.16} & \textbf{0.23 $\pm$ 0.15} \\
VAR-LASSO & 0.15 $\pm$ 0.02 & 0.50 $\pm$ 0.01 & 0.11 $\pm$ 0.00 \\
VAR-Ridge & 0.19 $\pm$ 0.00 & 0.49 $\pm$ 0.00 & 0.10 $\pm$ 0.00 \\
VAR-LiNGAM & 0.10 $\pm$ 0.05 & 0.49 $\pm$ 0.01 & 0.10 $\pm$ 0.00 \\
\bottomrule
\end{tabular}
\end{small}
\end{table}

\paragraph{Linear and count-process scaling at $n{=}80$.}
\label{app:n80_scaling}
Beyond the nonlinear tanh family above, we ran preliminary larger-scale benchmarks on the Linear VAR and Poisson GLM families at $n{=}80$ as part of the reviewer-response process. SBTG retains a non-trivial recovery signal in both: mean F1 of $0.160$ on the VAR family and $0.141$ on the Poisson family, averaged across the same noise/length grid as the tanh table above. By contrast, several alternative methods collapse to chance at this scale: VAR-LiNGAM achieves mean F1 $0.044$ (VAR) and $0.048$ (Poisson), and the deep-learning baselines NRI, NetFormer, and LINT remain near chance throughout. We do not claim that $n{=}10$ recovery extends without modification to the empirical regime; rather, SBTG remains feasible and continues to recover a modest but interpretable signal in a substantially more complex setting where competing methods do not.

\subsection{Cell Type Analysis}
\label{app:celltype}

\subsubsection{Cell Type Assignment}

Neurons are classified into functional categories using Cook's annotations:

\begin{table}[h!]
\centering
\caption{Cell type categories.}
\label{tab:celltypes}
\begin{tabular}{lrl}
\toprule
\textbf{Type} & \textbf{Count} & \textbf{Examples} \\
\midrule
Sensory (S) & 28 & ASE, AWA, AWC, ASH, ASK \\
Interneuron (I) & 34 & AIB, RIM, AVA, AVE, RIA \\
Motor (M) & 11 & RMD, SMD, RME, VA, VB \\
Other & 7 & GLR, DVA, PVC \\
\bottomrule
\end{tabular}
\end{table}

\subsubsection{Statistical Comparisons}

For each lag $r$ and cell-type pair $(A, B) \in \{S, I, M\}^2$, we extract significant edges $A \to B$ where $\{(j,i) : i \in A, j \in B, \mu_{r,ji} \text{ significant}\}$, compute mean $|\widehat{\mu}_{r,ji}|$ across these edges, perform within-lag comparison using Mann-Whitney $U$ test comparing $(A \to B)$ versus all other pairs at lag $r$, and conduct across-lag comparison using paired $t$-test for the same pair $(A \to B)$ at different lags. Heatmaps show mean $|\mu|$ for each $(A, B)$ pair with significance stars from Mann-Whitney tests.

\subsection{Monoamine Connectome Evaluation}
\label{app:monoamine}

\subsubsection{Alignment Procedure}

We preprocess neuron names (collapsing L/R and D/V subtypes), filter to the 80 neurons in our analysis set, construct adjacency matrices, and binarize such that a directed edge $(i \to j)$ exists if the transmitter/receptor is present.

\subsubsection{Metric Selection and Rationale}
We evaluate our model using three complementary metrics: AUROC (Area Under the Receiver Operating Characteristic Curve), AUPRC (Area Under the Precision-Recall Curve), and the Max F1 Score. The choice of these metrics is motivated by the extreme sparsity of the connectome ($<10\%$ density) and the need to capture different aspects of model performance.

\textbf{AUROC (Ranking Quality):} Measures the probability that a randomly chosen true edge is ranked higher than a randomly chosen non-edge. This metric effectively quantifies how well the model captures the underlying biological physics of the system, independent of the choice of threshold. A high AUROC indicates that the learned coupling scores assign statistically higher values to true biological connections.

\textbf{AUPRC (Reliability):} Measures the trade-off between precision and recall across all thresholds, focusing specifically on the positive class (edges). In sparse datasets like connectomes, AUROC can be overly optimistic due to the large number of true negatives. AUPRC provides a stricter measure of how reliable the top-ranked predictions are, penalizing false positives more heavily.

\textbf{F1 Score (Recoverability):} Represents the harmonic mean of precision and recall at the optimal decision threshold. While AUROC and AUPRC measure ranking performance across all thresholds, the Max F1 score answers a practical question: \textit{is there a single cut-off point that recovers a graph structure topologically similar to the ground truth?} We analyze the peaks of the F1 score specifically because they indicate the time lag at which the inferred functional graph is closest to the structural connectome.

\begin{figure}[ht]
    \centering
    \includegraphics[width=0.85\linewidth]{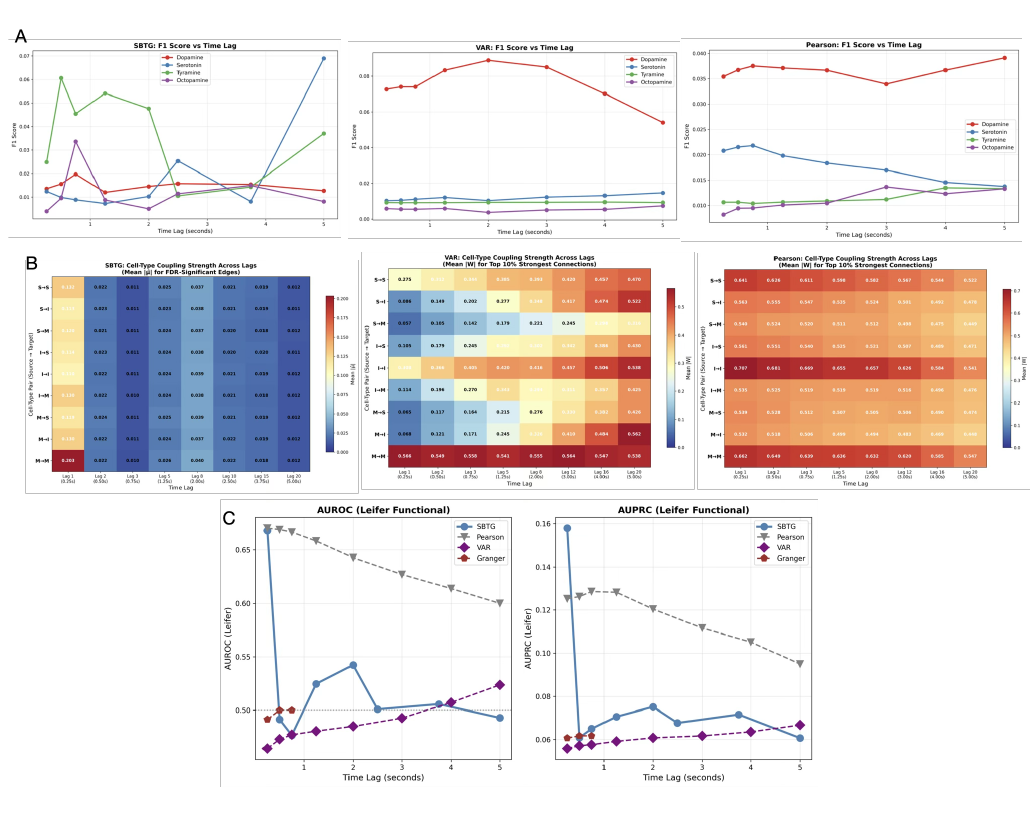}
    \caption{(A) Monoamine-connectome prediction (F1) versus lag for SBTG and representative baselines, shown separately by transmitter (dopamine, serotonin, tyramine, octopamine).(B) Cell-type coupling heatmaps (sensory/interneuron/motor) as a function of lag, comparing SBTG to VAR and Pearson. (C) Randi functional benchmark performance versus lag (AUROC, AUPRC) across methods.}
    \label{app:celltype-methods}
\end{figure}

The temporal profile of these recoverability scores is visualized in Figure~\ref{app:chem_gap} for Electrical (gap junctions) vs. Chemical (neurotransmitters) connectomes and ionotropic (GABA-A) or metabotropic (GABA-B) connectomes \cite{white1986structure, Cook2019, Yemini2021}.

\begin{figure}[ht]
    \centering
    \includegraphics[width=0.85\linewidth]{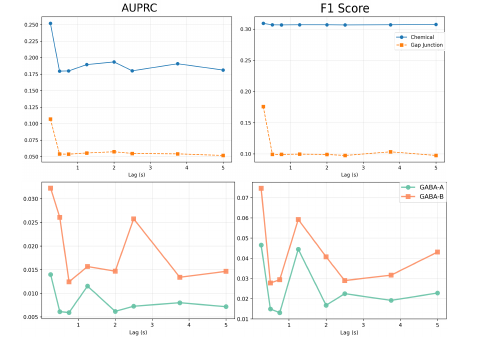}
    \caption{\textbf{Lag-resolved recoverability for structural and receptor networks.}
    \textbf{Top:} Chemical vs.\ Gap Junction benchmarks as a function of lag (left: AUPRC, right: F1).
    The Gap Junction network shows a sharp peak at $t=0.25$s and rapidly decays, whereas Chemical Synapses exhibit broader persistence across lags.
    \textbf{Bottom:} GABA-A vs.\ GABA-B benchmarks (left: AUPRC, right: F1), showing strongest recoverability at the shortest lag with additional lag-dependent structure, particularly for GABA-B.}
    \label{app:chem_gap}
\end{figure}

\subsection{Peak Lag Analysis}
\label{app:peak_analysis}

To characterize the dominant timescales of synaptic, functional, and neuromodulatory signaling, we analyzed the Model Performance ($F_1$ score) as a function of time lag.

To do so, first we identified the discrete lag $t_{\mathrm{max}}$ maximizing the observed $F_1$ score:
\begin{equation}
t_{\mathrm{max}} = \arg\max_{t \in \mathcal{T}} F_1(t).
\end{equation}
Second, to estimate the true biological peak $t_{\mathrm{peak}}$ between sampled points, we applied local parabolic interpolation. We fit a parabola to the discrete maximum $(t_{\mathrm{max}}, y_{\mathrm{max}})$ and its immediate neighbors $(t_{\mathrm{max}}-\Delta t, y_{\mathrm{prev}})$ and $(t_{\mathrm{max}}+\Delta t, y_{\mathrm{next}})$. The vertex of this parabola provides the sub-sample peak estimate:
\begin{equation}
t_{\mathrm{peak}} = t_{\mathrm{max}} + \frac{\Delta t}{2} \cdot \frac{y_{\mathrm{prev}} - y_{\mathrm{next}}}{y_{\mathrm{prev}} - 2y_{\mathrm{max}} + y_{\mathrm{next}}}.
\end{equation}
This estimator assumes the underlying curve is smooth and concave at the peak. Peaks occurring at the boundary of the sampling window (e.g., Serotonin at 5s) were not interpolated.

\begin{table}[h]
\centering
\caption{Peak time lags for functional connectivity benchmarks. Discrete peaks are the lags with maximum sampled F1 score. Interpolated peaks are estimated using parabolic interpolation around the discrete maximum to find the sub-sample peak location.}
\label{tab:peak_lags_appendix}
\begin{small}
\begin{tabular}{lrrrr}
\toprule
Network & Peak Lag (s) & Interp Lag (s) & Peak F1 & Interp F1 \\
\midrule
Monoamine: dopamine & 0.75 & 0.82 & 0.02 & 0.02 \\
Monoamine: serotonin & 5.00 & 5.00 & 0.07 & 0.07 \\
Monoamine: tyramine & 0.50 & 0.55 & 0.06 & 0.06 \\
Monoamine: octopamine & 0.75 & 0.87 & 0.03 & 0.04 \\
Structural: Cook & 0.25 & 0.25 & 0.34 & 0.34 \\
Functional: Randi & 2.00 & 1.88 & 0.69 & 0.69 \\
\bottomrule
\end{tabular}
\end{small}
\end{table}


\end{document}